\documentclass[11pt, a4paper]{article} 
\usepackage[english]{babel}
\usepackage[utf8]{inputenc}
\usepackage[a4paper]{geometry}

\usepackage{color}
\usepackage{graphicx}

\usepackage{mathrsfs}

\usepackage{float}

\usepackage{bbm}
\usepackage{enumerate}
\usepackage{tikz}
\usetikzlibrary{matrix}
\usepackage{pgfplots}
\usetikzlibrary{arrows.meta}

\usepackage[colorlinks=true,linkcolor=black,citecolor=black]{hyperref}

\usepackage[symbol]{footmisc}

\usepackage{amsmath}
\usepackage{amssymb}
\usepackage{amsthm}

\DeclareMathOperator*{\essinf}{ess\,inf}

\newcommand{\beao}{\begin{eqnarray*}}
\newcommand{\eeao}{\end{eqnarray*}\noindent}

\newcommand{\beam}{\begin{eqnarray}}
\newcommand{\eeam}{\end{eqnarray}\noindent}

\def\bbr{{\Bbb R}}   
\def\bbn{{\Bbb N}}

\newcommand{\eps}{{\varepsilon}}

\newcommand{\vt}{\vartheta}

\newcommand{\wh}{\widehat}
\newcommand{\wt}{\widetilde}

\theoremstyle{theorem}
\newtheorem{Theorem}{Theorem}[section]
\newtheorem{Lemma}[Theorem]{Lemma}
\newtheorem{Corollary}[Theorem]{Corollary}
\newtheorem{Proposition}[Theorem]{Proposition}
\newtheorem{definition}[Theorem]{Definition}
\newtheorem{Remark}[Theorem]{Remark}
\newtheorem{Example}[Theorem]{Example}

\numberwithin{equation}{section}

\title{Prospective strict no-arbitrage and the fundamental theorem of asset pricing under transaction costs\thanks{We would like to thank the editor, Prof. Martin Schweizer, an anonymous associate editor, and two anonymous referees for their valuable comments and suggestions from which the manuscript greatly benefited.}}
\author{Christoph Kühn\thanks{Institut für Mathematik, Goethe-Universität Frankfurt, D-60054 Frankfurt a.M., Germany, e-mail: \{ckuehn,
			molitor\}@math.uni-frankfurt.de}\and Alexander Molitor\footnotemark[2]}
\date{}

\begin{document}
	\maketitle
\begin{abstract}
	In discrete time markets with proportional transaction costs, Schachermayer~\cite{schachermayer2004fundamental} shows that {\em robust no-arbitrage} is equivalent to the existence of a strictly consistent price system.

	In this paper, we introduce the concept of {\em prospective strict
	no-arbitrage} that is a variant of the {\em strict no-arbitrage} property from Kabanov, R\'{a}sonyi, and Stricker~\cite{Kabanov2002}. The prospective strict no-arbitrage condition is slightly weaker than robust no-arbitrage, and it implies that the set of portfolios attainable from zero initial endowment is closed in probability. A weak version of 
prospective strict no-arbitrage	turns out to be equivalent to the existence of a consistent price system. In contrast to the fundamental theorem of asset pricing of  Schachermayer~\cite{schachermayer2004fundamental}, the consistent frictionless prices may lie on the boundary of the bid-ask spread.

On the technical level,	a crucial difference to Schachermayer~\cite{schachermayer2004fundamental} and Kabanov-R\'{a}sonyi-Stricker~\cite{kabanov2003closedness} is that we prove closedness without having at hand that the null-strategies form a linear space.
\end{abstract}

\begin{tabbing}

 
 


\end{tabbing}

\section{Introduction}

In frictionless finite discrete time financial market models, the absence of arbitrage opportunities is equivalent to the existence of an equivalent probability measure under which the discounted price processes are martingales. This result is called the fundamental theorem of asset pricing~(FTAP). In the case of a finite probability space, it goes back to the work of Harrison and 
Pliska~\cite{harrison1981martingales}. The extension to arbitrary probability spaces is known as the Dalang-Morton-Willinger Theorem \cite{dalang1990equivalent}, whose original proof was subsequently refined by several authors, see, e.g., \cite{schachermayer1992hilbert}, \cite{kabanov2001teacher}. In the later proofs, the implication that, in frictionless markets, the absence of arbitrage opportunities implies 
that the set of hedgeable claims attainable from zero endowment 
is closed in probability is identified as the key lemma. 

For a finite probability space, Kabanov and Stricker~\cite{kabanov2001harrison} extend the FTAP of Harrison and Pliska to models with  proportional transaction costs. 
They consider a general ``currency model'' with finitely many currencies (assets), which we also follow in the current paper. It allows to buy any asset by paying with any other asset. In this general framework, there need not exist an asset which can play the role of a bank account, i.e., an asset which can be involved in every transaction at minimal costs.
Kabanov and Stricker show that no-arbitrage~(NA) is equivalent to the existence
of a so-called consistent price system~(CPS), which is a multidimensional martingale under the objective probability measure taking values within the dual of the cone of solvent portfolios at each point in time. For infinite probability spaces, this equivalence fails:
Schachermayer~\cite{schachermayer2004fundamental} provides  
an example for an arbitrage-free market which allows for an approximate arbitrage and consequently a CPS cannot exist (see Example~3.1 therein). There arises the obvious question under which stronger no-arbitrage conditions the existence of a CPS can be guaranteed.
Schachermayer~\cite{schachermayer2004fundamental} introduces the  concept of robust no-arbitrage~(NA\textsuperscript{r}) -- a no-arbitrage condition which is robust with respect to small changes in the bid-ask spreads. Loosely speaking, if the bid-ask spread 
(of a pair of assets) 
does not vanish, there have to exist more favorable bid-ask prices, leading to a smaller spread, such that the modified market still satisfies (NA).
Schachermayer shows that (NA\textsuperscript{r}) implies that 
the set of hedgeable claims attainable from zero endowment, in the following denoted by $\mathcal{A}$, is closed in probability, and (NA\textsuperscript{r}) is equivalent to the existence of a strictly consistent price system~(SCPS), that is a martingale taking values within the relative interior of the dual of the cone of solvent portfolios at each point in time.

An alternative condition is the strict no-arbitrage~(NA\textsuperscript{s})~property introduced by Kabanov, R\'{a}sonyi, and Stricker~\cite{Kabanov2002}. Loosely speaking, a market model satisfies 
(NA\textsuperscript{s}) iff any claim which is attainable from zero endowment up to some intermediate time~$t$
and which can be liquidated in $t$ for sure, can also be attained from zero endowment by trading at time~$t$ only. (NA\textsuperscript{s}) alone does not imply the existence of a CPS (see Example~3.3 in \cite{schachermayer2004fundamental} for the existence of an approximate arbitrage under (NA\textsuperscript{s})), but together with the Penner-condition this implication holds
(see Penner~\cite{Penner} and Theorem~2 of Kabanov-R\'{a}sonyi-Stricker~\cite{kabanov2003closedness}). Loosely speaking, the Penner-condition postulates that any 
``free-round-trip'' of exchanging assets that can be carried out in the next period for sure --
given the information of the current period -- can already be carried out in the current period. Together with 
(NA\textsuperscript{s}), it allows to show that the so-called 
null-strategies, i.e., the increments of self-financing portfolio processes with vanishing terminal value,
form a linear space. This is also a crucial argument in \cite{schachermayer2004fundamental} to show closedness of $\mathcal{A}$, which is the main step to show the existence of a CPS. 
Indeed, it is shown by Rokhlin~\cite{rokhlin2008constructive} that the 
vector space property of null-strategies is equivalent to (NA\textsuperscript{r}).

A different approach to study 
the occurrence of an approximate arbitrage
is followed  in Jacka, Berkaoui, and Warren~\cite{jacka2008no}. They provide a necessary and sufficient condition for $\mathcal{A}$ to be closed in probability and construct adjusted trading prices such that the corresponding cone of hedgeable claims attainable from zero endowment either contains an arbitrage or corresponds to the closure of $\mathcal{A}$. Put differently, they 
postulate the (weak) no-arbitrage condition for an adjusted trading model instead of postulating a stronger no-arbitrage condition for the original one (cf. also Remark~\ref{26.11.2018.1} below).

Furthermore, it is important to note that the closedness of 
$\mathcal{A}$ is not necessary for the existence of a CPS.
In the case of only two assets 
(e.g., a bank account and one risky stock), it is shown by 
Grigoriev~\cite{grigoriev2005low} that (NA) already implies the existence of a 
CPS -- although $\mathcal{A}$ need not be closed (see Example~1.3 in \cite{grigoriev2005low} and Proposition~3.5 in L\'epinette and Zhao~\cite{jun.lepinette}
for the non-closedness of the set of attainable liquidation values).
This means that already in dimension two, additional conditions are required to guarantee that the set of attainable liquidation values is closed. For this,
L\'epinette and Zhao~\cite{jun.lepinette} provide an intuitive and easy to verify condition~(Condition~E) that takes the postponing of trades into consideration. Their proof uses the existence of a CPS that is guaranteed by Grigoriev~\cite{grigoriev2005low} in the case of an arbitrage-free model with two assets. 
On the other hand, already for three assets, there is a counterexample showing that (NA) does not imply the existence of a CPS (see Example~4.6 in \cite{kuhn2018local}).
The goal of the current paper is twofold: 
\begin{itemize}
\item We want to provide an (easy to interpret) no-arbitrage condition which is as weak as possible and under which the set~$\mathcal{A}$ of terminal portfolios attainable from zero endowment is closed. 
\item We want to establish a FTAP with CPSs which are not necessarily strict as in the FTAP of Schachermayer~\cite{schachermayer2004fundamental}.  
\end{itemize}
For this, we introduce a variant of (NA\textsuperscript{s}), that we call
\textit{prospective strict no-arbitrage}~(NA\textsuperscript{ps}) and that turns out to be sufficient to guarantee that $\mathcal{A}$ is closed in probability (see Theorem~\ref{theo:MainResult1}). We say that the market model satisfies (NA\textsuperscript{ps}) iff any claim which is attainable from zero endowment by trading up to some time~$t$ and which can {\em subsequently} be liquidated for sure, can also be attained from zero endowment in the {\em subsequent} periods (here, ``subsequent'' is not understood in a strict sense). This means that in contrast to the (NA\textsuperscript{s}) criterion, we do not distinguish between a trade that can be realized at time~$t$ and a trade from which we know at time~$t$ for sure that it can be realized in the future.  In the special case of efficient friction, (NA\textsuperscript{ps}) and (NA\textsuperscript{s}) are equivalent (see Proposition~\ref{prop:EF}).

In our proofs, we cannot rely on the vector space property of the null-strategies, which was central in the arguments of Schachermayer and Kabanov-R\'{a}sonyi-Stricker. Indeed, it was shown by Rokhlin~\cite{rokhlin2008constructive} that this property is equivalent to (NA\textsuperscript{r}), which is strictly stronger than  (NA\textsuperscript{ps}).
Our proof relies on a decomposition of the trading possibilities in ``reversible'' and ``purely non-reversible'' transactions at each point in time, where we call a transaction ``reversible'' if the resulting portfolio can be liquidated in the later periods for sure. 
This decomposition can be seen as a non-linear, only positively homogeneous generalization of the  projection on the set of null-strategies that is
used in the case that the null-strategies form a vector space.
Given a trading strategy, we then consider only the ``purely non-reversible'' part at each point in time and postpone the ``reversible'' part to later points in time, where more information is available. This is possible by (NA\textsuperscript{ps}), and, as it turns out, sufficient to assert that $\mathcal{A}$ is closed in probability. 
Consequently, (NA\textsuperscript{ps}) implies
the existence of a CPS. 

On the other hand, as described above, a CPS can exist although the set~$\mathcal{A}$ is not closed. Consequently, the existence of a CPS cannot be equivalent to (NA\textsuperscript{ps}). 
But, for a weak version of (NA\textsuperscript{ps}), called \textit{weak prospective strict no-arbitrage} (NA\textsuperscript{wps}),  
we have equivalence to the existence of a CPS (see Theorem~\ref{theo:MainResult2}). A market satisfies 
(NA\textsuperscript{wps}) iff there exists an at least as favorable market which satisfies (NA\textsuperscript{ps}). Since the second market need not be strictly more favorable than the original one,
(NA\textsuperscript{ps}) implies (NA\textsuperscript{wps}). 
Hence, we establish a FTAP, which complements those of Schachermayer~\cite{schachermayer2004fundamental} and Kabanov-R\'{a}sonyi-Stricker~\cite{Kabanov2002,kabanov2003closedness}. 
The main difference is that the resulting CPS may lie on the relative boundary of the bid-ask-spread. In Section~\ref{sec:Examples}, Figure~\ref{fig:NApNotNAr}  
illustrates a very simple example for this. Alternatively, 
one may think of an actually frictionless market that is written as a model with efficient friction in the following way. Each point in time is split into two points. Under the same information, at the first point, the investor can only buy, and at the second point she can only sell a stock. If the frictionless market satisfies (NA), the artificial market with friction has a CPS but not a SCPS. Thus, at least from a conceptual point of view, it is desirable to have a FTAP with arbitrary CPSs as well.

In the case of a finite probability space, (NA\textsuperscript{wps})
is equivalent to (NA), which means that our version of the FTAP can be seen as a generalization of the above mentioned FTAP by Kabanov and Stricker~\cite{kabanov2001harrison} (see part~2 of Theorem~1 therein) to the case of arbitrary  probability spaces.
Finally, we motivate the  (NA\textsuperscript{wps}) condition by an example which shows that (NA\textsuperscript{wps}) cannot be replaced by a further weakening
of the (NA\textsuperscript{ps}) condition (see Example~\ref{example:Strengthofthecondition}).

The remainder of the paper is organized as follows. In Section~\ref{3.11.2018.1}, we introduce the framework of financial modeling, the prospective strict no-arbitrage condition, and the weak prospective strict no-arbitrage condition. We  relate these properties to the robust no-arbitrage and the strict no-arbitrage condition and state the main results of the paper  (Theorem~\ref{theo:MainResult1} and Theorem~\ref{theo:MainResult2}). The proofs can be found in Section~\ref{22.11.2018.1}. In Section~\ref{sec:Examples}, there are two very simple examples that
illustrate the differences between the above mentioned no-arbitrage conditions and a more sophisticated example (Example~\ref{example:Strengthofthecondition}) that shows the effect of a possible ``cascade'' of approximate hedges.

\section{Prospective strict no-arbitrage and consistent price systems}\label{3.11.2018.1}
	We now introduce the market model and the relevant notation. We work on a probability space $(\Omega,\mathcal{F},\mathbb{P})$ equipped with a discrete time filtration $(\mathcal{F}_t)_{t=0}^T$, $T\in\bbn$, such that $\mathcal{F}_T=\mathcal{F}$. The space (of equivalence classes) of $\mathcal{F}_t$-measurable $d$-dimensional random vectors is denoted by $L^0(\mathbb{R}^d,\mathcal{F}_t)$. 
	For a set-valued mapping $\omega\mapsto N(\omega)\subseteq \mathbb{R}^d$, we denote by $L^0(N,\mathcal{F}_t):=\{v\in L^0(\mathbb{R}^d,\mathcal{F}_t)\mid v(\omega)\in N(\omega)\ \mathrm{for\ a.e.}\ \omega\in\Omega\}$ the set of $\mathcal{F}_t$-measurable selectors of $N$.
	As usual, the spaces are equipped with the topology of the convergence in probability, and we write $L^0(N):=L^0(N,\mathcal{F}_T)$.

We work with the market model with proportional transaction costs from 
Schachermayer~\cite{schachermayer2004fundamental}, where the reader may find a discussion about its economical meaning and its connection to the models of \cite{Kabanov2002} and \cite{kabanov2001harrison}. There are $d\in\bbn$ traded assets, and a $d\times d$-matrix $\Pi=(\pi^{ij})_{1\leq i,j\leq d}$ is called \textit{bid-ask matrix} if
	\begin{enumerate}[(i)]
	\item \label{item:DefBidAskMatrix1} $0<\pi^{ij}<\infty$, for $1\leq i,j\leq d$,
	\item \label{item:DefBidAskMatrix2} $\pi^{ii}=1$, for $1\leq i\leq d$,
	\item \label{item:DefBidAskMatrix3} $\pi^{ij}\leq \pi^{ik}\pi^{kj}$, for $1\leq i,j,k\leq d$.
	\end{enumerate}
	The terms of trade of the $d$ assets are specified by a \textit{bid-ask process} $(\Pi_t)_{t=0}^T$, i.e., an adapted $d\times d$--matrix-valued process such that for each $\omega\in \Omega$ and $t\in\{0,\dots,T\}$, $\Pi_t(\omega)$ is a bid-ask matrix. 
	For each $t\in\{0,\dots,T\}$, the random matrix $\Pi_t=(\pi^{ij}_t)_{1\leq i,j\leq d}$ specifies the exchanges available to the investor at time $t$. More precisely, the entry $\pi^{ij}_t$ denotes the number of units of asset $i$ for which an agent can buy one unit of asset $j$ at time $t$. Therefore, the \textit{set of portfolios attainable at zero endowment at time $t$}, which, in this context, consists of $\mathcal{F}_t$-measurable $\mathbb{R}^d$-valued random variables, is modeled by the convex cone 
	\begin{align}\label{eq:PortfoliosAttainable}
		\left\{\sum_{1\leq i,j\leq d}\lambda^{ij}(e^j-\pi^{ij}_te^i)-r\bigg \vert (\lambda^{ij})_{1\leq i,j\leq d}\in L^0(\mathbb{R}^{d\times d}_+,\mathcal{F}_t),\ r\in L^0(\mathbb{R}^d_+,\mathcal{F}_t)\right\},
	\end{align}
	where $e^i$ denotes the $i$-th unit vector of $\mathbb{R}^d$. This means that each portfolio is the result of an order $\lambda=(\lambda^{ij})_{1\leq i,j\leq d}\in L^0(\mathbb{R}^{d\times d}_+,\mathcal{F}_t)$, where $\lambda^{ij}$ denotes the units of assets $j$ ordered in exchange for asset $i$, and some non-negative amount $r\in L^0(\mathbb{R}^d_+,\mathcal{F}_t)$, which corresponds to the decision of the investor to ``throw away'' some non-negative quantities of each assets. Next, we define for each $\omega\in \Omega$ the polyhedral cone 
	\begin{align*}
	 -K(\Pi_t(\omega)):=\mathrm{cone}\left(\left\{e^j-\pi^{ij}_t(\omega)e^i\right\}_{1\leq i,j\leq d}, \left\{-e^i\right\}_{1\leq i\leq d}\right),
	\end{align*}
	which we abbreviate as $-K_t(\omega):=-K(\Pi_t(\omega))$.
	In Lemma~\ref{Lemma:Representations} below, we briefly verify the intuitively obvious fact that the set given in \eqref{eq:PortfoliosAttainable} coincides with the set $L^0(-K_t,\mathcal{F}_t)$ of $\mathcal{F}_t$-measurable selectors of the set-valued mapping $\omega\mapsto -K_t(\omega)$. We use this equality throughout the paper and refer to $L^0(-K_t,\mathcal{F}_t)$ as the set of portfolios attainable from zero endowment at time~$t$. 

	 \begin{definition}
		 An $\mathbb{R}^d$-valued adapted process $\vartheta=(\vartheta_t)_{t=0}^T$ is called self-financing portfolio process for the bid-ask process $(\Pi_t)_{t=0}^T$ if
		 \begin{align}
			 \vartheta_t-\vartheta_{t-1}\in L^0(-K_t,\mathcal{F}_t)\quad\mbox{for all }t=0,\ldots,T,
		 \end{align}
		 where $\vartheta_{-1}:=0$. Consequently, for each pair $(s,t)$ with $s,t\in\{0,\ldots,T\}$ and $s\le t$, the convex cone of hedgeable claims attainable from zero endowment between $s$ and $t$ is denoted by $\mathcal{A}_s^t$ and is defined to be
		 \begin{align*}
			 \mathcal{A}_s^t:=\sum_{k=s}^{t}L^0(-K_k,\mathcal{F}_k).
		 \end{align*}
For an alternative bid-ask process~$(\wt{\Pi}_t)_{t=0}^T$, the corresponding set is denoted by $\wt{\mathcal{A}}_s^t$, where $-\wt{K}_t(\omega):=-K(\wt{\Pi}_t(\omega))$ for all $\omega\in\Omega$ and $t=0,\dots,T$.
\end{definition}
	 The primary object of interest in this paper is the cone $\mathcal{A}_0^T$ of hedgeable claims attainable from zero endowment between $0$ and $T$. However, we still need the following auxiliary notions. Let $K_t(\omega):=-(-K_t(\omega))$ for each $\omega\in\Omega$, then the convex cone $L^0(K_t,\mathcal{F}_t)$ is called the \textit{set of solvent portfolios at time $t$} and the (polyhedral) cone $K_t(\omega)$ is called the \textit{solvency cone} corresponding to the bid-ask matrix $\Pi_t(\omega)$. Indeed, for each portfolio $v\in L^0(K_t,\mathcal{F}_t)$ the portfolio $-v\in L^0(-K_t,\mathcal{F}_t)$ is attainable at price zero, thus the portfolio $v$ can be liquidated to zero and, consequently, is solvent. Similarly, let $K^0_t(\omega):=K_t(\omega)\cap -K_t(\omega)$ for each $\omega\in \Omega$, then $L^0(K^0_t,\mathcal{F}_t)$ denotes the space of portfolios, which are attainable at zero endowment, and vice versa, are solvent.

Before we introduce our new no-arbitrage condition, we recall the concepts of no-arbitrage
from the literature.
\begin{definition}[compare to \cite{schachermayer2004fundamental} and \cite{kabanov2003closedness}]\leavevmode 
\label{def:NAlit}
	 	\begin{enumerate}[(i)]
	 		\item The bid-ask process $(\Pi_t)_{t=0}^T$ satisfies the no-arbitrage property  (NA) if
	 		\begin{align}
		 		\label{def:NA} \mathcal{A}_0^T\cap L^0(\mathbb{R}^d_+)=\{0\}.
	 		\end{align}
	 		\item The bid-ask process $(\Pi_t)_{t=0}^T$ satisfies the strict no-arbitrage property (NA\textsuperscript{s}) if 
	 		\begin{align}
		 		\label{def:NAs} \mathcal{A}_0^t\cap L^0(K_t,\mathcal{F}_t)\subseteq L^0(K^0_t,\mathcal{F}_t)\quad\mbox{for all }t=0,\ldots,T.
	 		\end{align}
	 		\item The bid-ask process $(\Pi_t)_{t=0}^T$ satisfies the robust no-arbitrage condition (NA\textsuperscript{r}) if there is a bid-ask process $(\widetilde{\Pi}_t)_{t=0}^T$ with smaller bid-ask spreads in the sense that 
\beam\label{9.11.2018.1}
\mbox{the spread $\left[\frac{1}{\widetilde{\pi}_t^{ji}(\omega)},\widetilde{\pi}_t^{ij}(\omega)\right]$ is contained in the relative interior of $\left[\frac{1}{\pi^{ji}_t(\omega)},{\pi_t}^{ij}(\omega)\right]$}
\eeam 		
for all $1\leq i,j\leq d$, $t\in \{0,\dots,T\}$ and almost all $\omega\in \Omega$, 
such that $(\widetilde{\Pi}_t)_{t=0}^T$ satisfies the no-arbitrage condition (NA).
	 	\end{enumerate}
	 \end{definition}
We just note that in the case of vanishing bid-ask spreads, condition (\ref{9.11.2018.1}) is satisfied by the choice of $\wt{\Pi}=\Pi$, i.e., frictionless markets are not excluded.

It is well known that although each of the cones $L^0(-K_t,\mathcal{F}_t)$ is closed with regard to the convergence in probability, the cone $\mathcal{A}_0^T$ may fail to be closed. As already mentioned, neither (NA) nor (NA\textsuperscript{s}) are strong enough to guarantee that $\mathcal{A}_0^T$ is closed (see Examples 3.1 and 3.3 in \cite{schachermayer2004fundamental}). This is in contrast to the frictionless case, where (NA) is sufficient (see, e.g., Theorem 6.9.2 in \cite{delbaen2006mathematics}). In the present context Schachermayer~\cite{schachermayer2004fundamental} showed that the robust no-arbitrage condition (NA\textsuperscript{r}) is strong enough to assure that $\mathcal{A}_0^T$ is closed. We now introduce a slight weakening of (NA\textsuperscript{r})  called prospective strict no-arbitrage (NA\textsuperscript{ps}), which is still sufficient to guarantee that $\mathcal{A}_0^T$ is closed.
	 \begin{definition} \label{def:NoProspectiveArbitrage}
		The bid-ask process $(\Pi_t)_{t=0}^T$ satisfies the prospective strict no-arbitrage property~(NA\textsuperscript{ps}) if 
 		\begin{align*}
	 		\mathcal{A}_0^t\cap(-\mathcal{A}_t^T)\subseteq \mathcal{A}_t^T\quad\mbox{for all }t=0,\ldots,T.
 		\end{align*}
	 \end{definition}
	 \begin{Remark}
	 	The (NA\textsuperscript{ps}) property has the following interpretation: any claim $v\in \mathcal{A}_0^t$ attained by trading up to time~$t$ which can be reduced to the zero portfolio in $t$ or in the subsequent periods, i.e., $-v\in \mathcal{A}_t^T$, has to be attainable by trading between $t$ and $T$ only, i.e., $v\in \mathcal{A}_t^T$. It is a variant of the 
	 	(NA\textsuperscript{s}) condition, that postulates that any claim $v\in \mathcal{A}_0^t$ which can be liquidated at time~$t$, i.e., $-v\in \mathcal{A}_t^t$, has to be attainable at time~$t$ as well, i.e., $v\in \mathcal{A}_t^t$. The only difference is that we do not distinguish between a trade at time $t$ and a trade from which one knows for sure 
	 	at time~$t$	that it can be realized in the future. 
	 	
	 	Put differently, for every $t$, we review the trading up to time~$t$. 
	 	Either one does not gain advantage from the trading 
	 	since the same terminal position can be achieved for sure by starting to trade at $t$. Or, one takes some risk by the trading up to time~$t$ since the position cannot be liquidated for sure in the future.
	 \end{Remark}
	 The conditions (NA\textsuperscript{ps})
	 and (NA\textsuperscript{s}) coincide under efficient friction (see Proposition~\ref{prop:EF}). We can already formulate the first main result of the paper:
	 \begin{Theorem} \label{theo:MainResult1}
		 If the bid-ask process $(\Pi_t)_{t=0}^T$ has the prospective strict no-arbitrage property~(NA\textsuperscript{ps}), then the convex cone $\mathcal{A}_0^T$ is closed with regard to the convergence in probability.
	 \end{Theorem} 
	The theorem above has obvious consequences for the existence of dual variables. 
For a given bid-ask matrix $\Pi$, the (positive) dual cone $K^\star$ of the solvency cone $K=K(\Pi)$ is defined by $K^\star:=\{w\in\mathbb{R}^d: \langle v,w\rangle \geq 0 \ \text{for all}\ v\in K\}$. 
For the bid-ask process~$(\Pi_t)_{t=0}^T$, this induces the set-valued process~$(K^\star_t)_{t=0}^T$ of dual cones. We can now define the notion of \textit{consistent price systems}, which is dual to the notion of self-financing portfolio process and plays a similar role as the notion of an equivalent martingale measure in the frictionless theory. Once again for a detailed discussion of the economical interpretation we refer to \cite{schachermayer2004fundamental}. 
\begin{definition}
	 An adapted $\mathbb{R}^d_+$-valued process $Z=(Z_t)_{t=0}^T$ is called a consistent price system~(CPS) for the bid-ask process  $(\Pi_t)_{t=0}^T$ if $Z$ is a martingale under $\mathbb{P}$ and $Z_t\in L^0(K^\star_t\setminus\{0\},\mathcal{F}_t)$, i.e., $Z_t(\omega)\in K^\star_t(\omega)\setminus\{0\}$ for a.e. $\omega \in \Omega$ and each $t\in\{0,\dots,T\}$.
	 \end{definition}
	 We have the following consequence of Theorem~\ref{theo:MainResult1}.
	 \begin{Corollary} \label{cor:MainResult1}
		 If the bid-ask process $(\Pi_t)_{t=0}^T$ satisfies the prospective strict no-arbitrage condition (NA\textsuperscript{ps}), then it admits a consistent price 
		 system~(CPS). More generally, for any given strictly positive $\mathcal{F}_T$-measurable function $\varphi:\Omega\to(0,1]$, there is a CPS~$Z=(Z_t)_{t=0}^T$ with $\Vert Z_T\Vert_2\leq M\varphi$ a.s. for some $M\in\bbr_+\setminus\{0\}$, where $\Vert \cdot\Vert_2$ denotes the Euclidean norm on $\mathbb{R}^d$.  
	 \end{Corollary}
	 \begin{Remark}
	 	An abstract version of (NA\textsuperscript{ps}) reads:
	 	If a strategy up to time $t$ can be extented to a strategy without losses at $T$, then any other extension beyond $t$ can be dominated at $T$ by a strategy that does not trade before $t$.
	 	
	 	This scheme can be formalized in a quite canonical way in diverse market models including, e.g., capital gains taxes, uncertainty about the execution of limit orders,
	 	or dividend paying assets, where the basic problem from Example~3.1 in Schachermayer~\cite{schachermayer2004fundamental}, can also occur (see, e.g., Example~4.5 in \cite{kuhn2018local}). The arguments of our proofs may be adapted to these models to show that the set of attainable terminal portfolios is closed.
	 	
	 	For example, in the context of optimal investment problems with utility functions on the positive real line, this means, roughly speaking, that the set $\mathcal{C}$ of non-negative random variables dominated by the liquidation value of an attainable portfolio (with a given initial endowment) is also closed in probability. Hence, defining the set of dual variables $\mathcal{D}$ as the polar set of $\mathcal{C}$, the abstract versions of the duality results in Kramkov and Schachermayer~\cite[Theorem 3.1 and 3.2]{kramkov1999} may also be applied to these models. 
	 \end{Remark} 
	The opposite of Corollary~\ref{cor:MainResult1} fails to be true. More generally,
	by Example 1 in Section 3.2.4 of \cite{kabanov2009markets}, there cannot exist a  no-arbitrage criterion that both guarantees closedness of $\mathcal{A}^T_0$ and that is equivalent to the existence of a CPS, cf. also the discussion in Remark~\ref{2.4.2019.1}. We can however establish an equivalence if we pass from (NA\textsuperscript{ps}) to a weaker notion of prospective strict no-arbitrage.
	 \begin{definition}\label{11.11.2018.1}
		 The bid-ask process $(\Pi_t)_{t=0}^T$ satisfies the weak prospective strict no-arbitrage property (NA\textsuperscript{wps}) if there is a bid-ask process $(\widetilde{\Pi}_t)_{t=0}^T$ with $\widetilde{\Pi}_t\leq \Pi_t$ a.s. for all $t=0,\dots,T$, such that $(\widetilde{\Pi}_t)_{t=0}^T$ satisfies the prospective strict no-arbitrage condition~(NA\textsuperscript{ps}).
	 \end{definition}
	 The (NA\textsuperscript{wps}) condition is obviously a weakening of the 
(NA\textsuperscript{ps}) condition since the bid-ask process $(\wt{\Pi}_t)_{t=0}^T$ in Definition~\ref{11.11.2018.1} need not be strictly more favorable than $(\Pi_t)_{t=0}^T$. The difference between the two conditions is illustrated in Example~\ref{example:WPros_Pros} below, see also Remark~\ref{2.4.2019.1}. Our second main result is the following fundamental theorem of asset pricing.
	 \begin{Theorem}\label{theo:MainResult2}
		 A bid-ask process $(\Pi_t)_{t=0}^T$ satisfies the weak prospective strict no-arbitrage condition~(NA\textsuperscript{wps}) if and only if it admits a consistent price system~(CPS). 
	 \end{Theorem}
	\begin{Remark}
Theorem~\ref{theo:MainResult2} extends part~2 of Theorem~1 in
Kabanov and Stricker~\cite{kabanov2001harrison} to the case of infinite
probability spaces. Combining these two theorems, it can be seen that 
(NA\textsuperscript{wps}) possesses the nice property 
that it is equivalent to (NA) if $|\Omega|<\infty$.
	\end{Remark}	 
\begin{Remark}
In addition, (NA) and (NA\textsuperscript{wps}) coincide in the case of only two assets on arbitrary probability spaces, which follows from the equivalence of (NA) and the existence of a CPS, derived by Grigoriev~\cite{grigoriev2005low}. 
\end{Remark}
\begin{Remark}\label{2.4.2019.1}
	In the following discussion, we identify an ``absence of arbitrage'' criterion~$\mathcal{C}$ with the set of bid-ask processes which satisfy the criterion
	and call it monotone if for all bid-ask processes $\wt{\Pi}\le \Pi$, $\wt{\Pi}\in\mathcal{C}$ implies that $\Pi\in\mathcal{C}$.
	Monotonicity is obviously satisfied by the simple (NA) condition. 
	The more sophisticated criteria (NA\textsuperscript{s}), (NA\textsuperscript{r}), and (NA\textsuperscript{ps})
	are in general only monotone if bid-ask matrices without efficient friction are excluded from the consideration, i.e., $\pi^{ij}\pi^{ji}\ge  \wt{\pi}^{ij}\wt{\pi}^{ji}>1$ for all $i\not=j$.                
	On the one hand, the equivalence to the existence of a CPS can only hold for a monotone criterion. 
	On the other hand, the closedness of the set of attainable portfolios 
	does not transfer to a market with a less favorable bid-ask process 
	(see, e.g., Examples \cite[Section 3.2.4, Example 1]{kabanov2009markets} and \cite[Example 2.1]{jacka2008no}). Thus, to guarantee closedness, e.g., in the context of optimal investment problems, the limitation to monotone criteria would be unnecessarily restrictive.
	
	The (NA\textsuperscript{wps}) criterion can be characterized as 
	the ``strongest monotone criterion which is weaker than (NA\textsuperscript{ps})'', i.e., it follows directly from Definition~\ref{11.11.2018.1} that 
	\beam\label{6.3.2019.1}
	(NA^{wps}) = \bigcap_{(NA^{ps})\subseteq \mathcal{C},\ \mathcal{C}\mbox{\ is monotone }}\mathcal{C}.
	\eeam
	In the special case of a frictionless market, the criteria (NA\textsuperscript{ps}) and (NA\textsuperscript{wps})
	coincide (see Proposition~\ref{prop:Comparision}). 
	
	We stress that the picture cannot be as clear-cut as in the frictionless case. 
	In discrete time frictionless markets, (NA) already implies closedness (see Schachermayer~\cite{schachermayer1992hilbert}). In continuous time frictionless markets, Delbaen and Schachermayer~\cite{delbaen1994general} derived closedness in the appropriate topology under the economic meaningful assumption of ``no free lunch with vanishing risk'' (NFLVR), that is also necessary for the existence of an equivalent martingale measure. 
	Under transaction costs, the FTAP of  Delbaen and 
	Schachermayer~\cite{delbaen1994general}
	cannot hold. Namely, Example~3.1 in 
	Schachermayer~\cite{schachermayer2004fundamental} satisfies 
	(NFLVR) defined for multivariate portfolio processes, i.e., there does not exist an approximate arbitrage with short positions in any asset
	converging uniformly to zero, but a CPS does nevertheless not exist. 
\end{Remark}
The (NA\textsuperscript{wps}) property is not sufficient to assure that $\mathcal{A}_0^T$ is closed in probability. However, we have the following obvious consequence of Theorem~\ref{theo:MainResult2}.
	\begin{Corollary} \label{cor:MainResult2}
		If a bid-ask process $(\Pi_t)_{t=0}^T$ satisfies the weak prospective strict no-arbitrage condition (NA\textsuperscript{wps}), then we have $\overline{\mathcal{A}^{T}_0}\cap L^0(\mathbb{R}^d_+)=\{0\}$.
	\end{Corollary}
The short proof is also deferred to Section~\ref{22.11.2018.1}. 	
\begin{Remark}\label{20.11.2018.2}
(NA\textsuperscript{wps}) postulates the existence of a bid-ask process $(\widetilde{\Pi}_t)_{t=0}^T$ such that $\widetilde{\Pi}_t\leq \Pi_t$ a.s. for all $t=0,\dots,T$ and
		\begin{align}\label{11.11.2018.2}
			\widetilde{\mathcal{A}}_0^t\cap\left(-\widetilde{\mathcal{A}}_t^T\right)\subseteq\widetilde{\mathcal{A}}_t^T\quad\mbox{for all }t=0,\dots,T.
		\end{align}
		 One may ask if one can replace this condition with the following slightly weaker condition:  there exists a bid-ask process  $(\widetilde{\Pi}_t)_{t=0}^T$ satisfying (NA), such that $\widetilde{\Pi}_t\leq \Pi_t$ a.s. for all $t=0,\dots,T$ and
		\beam\label{10.11.2018.1}
			\mathcal{A}_0^t\cap\left(-\widetilde{\mathcal{A}}_t^T\right)\subseteq\widetilde{\mathcal{A}}_t^T\quad\mbox{for all }t=0,\dots,T
\eeam
Indeed, by Proposition~\ref{prop:Comparision}, (\ref{11.11.2018.2}) implies that $\wt{\Pi}$ satisfies (NA), which means that the second condition is a weakening of the first one. 
In condition~(\ref{10.11.2018.1}), the position at time~$t$ is achieved by trading in the original market, only its ``evaluation'' is made in the more favorable market model~$\wt{\Pi}$.   
But, maybe surprisingly, it turns out that (\ref{10.11.2018.1})  does not exclude the existence of an approximate arbitrage and thus a CPS need not exist (see Example~\ref{example:Strengthofthecondition} below).	
\end{Remark}

\begin{Proposition}\label{prop:Comparision}
	 	We have the following implications
	 	\begin{align}\label{eq:Implications}
		 	(NA^r)\Rightarrow (NA^{ps})\Rightarrow (NA^{wps})\Rightarrow (NA).
	 	\end{align}
	 \end{Proposition}
\begin{Remark}	 
All implications in (\ref{eq:Implications}) are strict (see 
Examples~\ref{example:WPros_Pros_Robust} and \ref{example:WPros_Pros} below;  
for (NA) $\not\Rightarrow$  (NA\textsuperscript{wps}), consider an arbitrage-free model with an approximate arbitrage, Example 3.1 in \cite{schachermayer2004fundamental},
and apply Corollary~\ref{cor:MainResult2}).
\end{Remark}	
	\begin{proof}[Proof of Proposition~\ref{prop:Comparision}]
Ad 	(NA\textsuperscript{r}) $\Rightarrow$ (NA\textsuperscript{ps}). Assume that the bid-ask process $(\Pi_t)_{t=0}^T$ satisfies (NA\textsuperscript{r}) and let $v\in \mathcal{A}_0^t$  such that $-v\in \mathcal{A}_t^T$. We have to show $v\in\mathcal{A}_t^T$. According to our assumption, we have  $v=\sum_{s=0}^{t}\widetilde{\xi_s}$ with $\widetilde{\xi}_s\in L^0(-K_s,\mathcal{F}_s)$ for $s=0,\ldots,t$ and  
$-v=\sum_{s=t}^T\widehat{\xi}_s$ with $\widehat{\xi}_s\in L^0(-K_s,\mathcal{F}_s)$ for $s=t,\ldots,T$. Hence, we define $\xi_s\in L^0(-K_s,\mathcal{F}_s)$ by
	 	\begin{align*}
		 	\xi_s:=\begin{cases} \widetilde{\xi}_s, & s<t,\\
								 \widetilde{\xi}_s+\widehat{\xi}_s, & s=t,\\
								 \widehat{\xi}_s, & s>t,
				 	\end{cases}
	 	\end{align*}
	 	and notice that $\sum_{s=0}^{T}\xi_s=v-v=0$. From Lemma~3.2.12 in \cite{kabanov2009markets}, it follows that $\xi_s\in L^0(K^0_s,\mathcal{F}_s)$ for all $s=0,\dots,T$. In particular, we have $\widehat{\xi}_s\in L^0(K_s,\mathcal{F}_s)$ for $s> t$. In addition, we have 
$\widehat{\xi}_t=-\widetilde{\xi}_t+\xi_t\in  L^0(K_t,\mathcal{F}_t) + L^0(K_t,\mathcal{F}_t) = L^0(K_t,\mathcal{F}_t)$.
This implies $v=\sum_{s=t}^T(-\wh{\xi}_s)\in \mathcal{A}_t^T$, which concludes the proof of the first implication.\\
	 	
Ad 	(NA\textsuperscript{ps}) $\Rightarrow$ (NA\textsuperscript{wps}). Obvious.\\

Ad 	(NA\textsuperscript{wps}) $\Rightarrow$ (NA). Assume that the bid-ask process $(\Pi_t)_{t=0}^T$ satisfies (NA\textsuperscript{wps}), i.e., there exists a bid-ask process $(\widetilde{\Pi}_t)_{t=0}^T$ with $\widetilde{\Pi}_t\leq \Pi_t$ a.s. for all $t=0,\dots,T$ and $(\widetilde{\Pi}_t)_{t=0}^T$ satisfies (NA\textsuperscript{ps}). 
Let $v\in \mathcal{A}_0^T\cap L^0(\mathbb{R}^d_+)\subseteq \wt{\mathcal{A}}_0^T\cap L^0(\mathbb{R}^d_+)$. This obviously implies that 
$-v\in L^0(-\wt{K}_T,\mathcal{F}_T)$ and hence, by (NA\textsuperscript{ps}) of $\wt{\Pi}$,  $v\in L^0(-\wt{K}_T,\mathcal{F}_T)$.  
Together with $(-\wt{K}_T(\omega))\cap \mathbb{R}^d_+=\{0\}$ for each $\omega\in\Omega$, which holds by the properties of a bid-ask matrix, this implies $v=0$ a.s. Thus, the bid-ask process $(\Pi_t)_{t=0}^T$ satisfies (NA).
\end{proof}
\begin{Proposition}\label{prop:EF} Let the efficient friction (EF) condition hold, i.e., \begin{align}\tag{EF}\label{eq:EF}
	K^0_t(\omega):=K_t(\omega)\cap\left(-K_t\left(\omega\right)\right)=\{0\} \quad \text{for all}\ t=0,\dots,T\ \text{and}\ \omega\in\Omega,
	\end{align} or, equivalently, $\pi^{ij}_t(\omega)\pi^{ji}_t(\omega)>1$ for all $1\leq i\neq j\leq d$, $t=0,\dots,T$, and $\omega\in\Omega$. Then, we have the equivalence
	\begin{align}\label{eq:Equivalence}
	(NA^{ps})\Leftrightarrow (NA^{s}).
	\end{align}
\end{Proposition}
\begin{Remark}
In general, (NA\textsuperscript{s}) is neither necessary nor sufficient for (NA\textsuperscript{ps}). Indeed, (NA\textsuperscript{ps}) $\not\Rightarrow$ (NA\textsuperscript{s}) is straightforward and 
(NA\textsuperscript{s}) $\not\Rightarrow$ (NA\textsuperscript{ps}) follows from Example~3.3 in \cite{schachermayer2004fundamental}. 	
	
But, it is also well-known that under efficient friction (NA\textsuperscript{r}) and (NA\textsuperscript{s}) are equivalent (cf. Theorem 1 in \cite{Kabanov2002} and Theorem 1.7 in \cite{schachermayer2004fundamental}). Thus, in this case (NA\textsuperscript{r}), (NA\textsuperscript{s}), and (NA\textsuperscript{ps}) coincide.
\end{Remark}
\begin{proof}[Proof of Proposition~\ref{prop:EF}]
In view of Proposition \ref{prop:Comparision} and the preceding remark, it is sufficient to show (NA\textsuperscript{ps})$\Rightarrow$(NA\textsuperscript{s}). Hence, we assume that (NA\textsuperscript{ps}) holds. Let us show by a backward induction on $t=T,T-1,\dots,0$ that $\mathcal{A}_0^t\cap L^0(K_t,\mathcal{F}_t)=\{0\}$. Let $t=T$ and $v\in \mathcal{A}_0^T\cap L^0(K_T,\mathcal{F}_T)$, then (NA\textsuperscript{ps}) implies $v\in L^0(-K_T,\mathcal{F}_T)$, i.e., $v\in L^0(K_T\cap (-K_T),\mathcal{F}_T)$, which, under \eqref{eq:EF}, is tantamount to $v=0$ a.s.
	
For the induction step $t+1\leadsto t$, we let $t<T$ and assume $A_0^s\cap L^0(K_s,\mathcal{F}_s)=\{0\}$ for $s=t+1,\dots,T$. Given $v\in \mathcal{A}_0^t\cap L^0(K_t,\mathcal{F}_t)$, we may write $v=\sum_{s=t}^{T}\xi_s$ for $\xi_s\in L^0(-K_s,\mathcal{F}_s)$ by (NA\textsuperscript{ps}). 
Since $-v\in L^0(-K_t,\mathcal{F}_t)$ and $-v+\sum_{s=t}^{T-1}\xi_s=-\xi_T$, we obtain $-v+\sum_{s=t}^{T-1}\xi_s\in\mathcal{A}_0^T\cap L^0(K_T,\mathcal{F}_T)$.
Thus, by the induction hypothesis, $-v+\sum_{s=t}^{T-2}\xi_s=-\xi_{T-1}$. Hence, $-v+\sum_{s=t}^{T-2}\xi_s\in\mathcal{A}_0^{T-1}\cap L^0(K_{T-1},\mathcal{F}_{T-1})$ and, again by the induction hypothesis, $-v+\sum_{s=t}^{T-3}\xi_s=-\xi_{T-2}$. Continuing inductively, we get $-v+\xi_t=0$, but this means $v\in L^0(K_t\cap (-K_t),\mathcal{F}_t)$ and thus $v=0$ a.s. by \eqref{eq:EF}.
\end{proof}
\begin{Remark}[Superhedging]
With Theorem~\ref{theo:MainResult1}, the superhedging result in 
Schachermayer~\cite{schachermayer2004fundamental} (see Theorem 4.1 therein) and its proof hold one-to-one under the slightly weaker assumption that  $(\Pi_t)_{t=0}^T$ satisfies (NA\textsuperscript{ps}) instead of 
(NA\textsuperscript{r}) -- only without the statement with ``strictly consistent price systems'' in the brackets.
\end{Remark}

\section{Proofs of the main results}\label{22.11.2018.1}
	This section is devoted to the proof of Theorem \ref{theo:MainResult1}. The other results of Section~\ref{3.11.2018.1} are standard consequences of $\mathcal{A}_0^T$ being closed and thus we 
mainly refer to the known results in the literature and highlight the minor adjustments.
The latter is postponed to the end of the section.
	
The main hurdle in the proof of Theorem~\ref{theo:MainResult1} is that we do not have at hand that the null-strategies, i.e., the elements of $(\xi_0,\dots,\xi_T)\in L^0(-K_0,\mathcal{F}_0)\times \dots\times L^0(-K_T,\mathcal{F}_T)$ with $\sum_{t=0}^{T}\xi_t=0$ a.s., form a linear space.  Namely, it is shown by Rokhlin~\cite{rokhlin2008constructive} that the implication
	 	\begin{align}
		 	\label{eq:NArImplication}
		 	\sum_{t=0}^{T}\xi_t=0 \ a.s. \ \text{with}\ \xi_t\in L^0(-K_t,\mathcal{F}_t)\Rightarrow \xi_t\in L^0(K^0_t,\mathcal{F}_t)\ \text{for all}\ t=0,\dots,T.
	 	\end{align}
is equivalent to (NA\textsuperscript{r}), which is strictly stronger than 
(NA\textsuperscript{ps}).
Thus, in the following we propose a new proof method which overcomes this hurdle.
	
Before starting with the main proof, we show that $L^0(-K_t,\mathcal{F}_t)$ coincides with the set given in \eqref{eq:PortfoliosAttainable}. Later, this allows us to argue directly with orders $\lambda\in L^0(\mathbb{R}^{d\times d}_+,\mathcal{F}_t)$ and vectors $r\in L^0(\mathbb{R}^d_+,\mathcal{F}_t)$ instead of the resulting elements of  $L^0(-K_t,\mathcal{F}_t)$.  In order to ease notation, we define for all $t=0,\dots,T$ the mapping $L_t:L^0(\mathbb{R}^{d\times d}_+,\mathcal{F}_t)\to L^0(\mathbb{R}^d,\mathcal{F}_t)$ by 
	\begin{align*}
		L_t(\lambda_t)=\sum_{1\leq i,j\leq d}\lambda^{ij}_t(e^j-\pi^{ij}_te^i)\quad
		\mbox{for all } \lambda_t=(\lambda^{ij}_t)_{1\leq i,j\leq d}\in L^0(\mathbb{R}^{d\times d}_+,\mathcal{F}_t).
	\end{align*}
	\begin{Lemma}\label{Lemma:Representations}
		Let $\Pi=(\Pi_t)_{t=0}^T$ denote a bid-ask process. Then, we have 
		\begin{align}\label{eq:Representation1}
			L^0(-K_t,\mathcal{F}_t)=\left\{L_t(\lambda_t)-r_t\mid \lambda_t\in L^0(\mathbb{R}^{d\times d}_+,\mathcal{F}_t), r_t\in L^0(\mathbb{R}^d_+,\mathcal{F}_t)\right\}
		\end{align}
		for all $t=0,\dots,T$ and, consequently, we have 
		\begin{align} \label{eq:Representation2}
			\mathcal{A}_s^t=\left\{\sum_{k=s}^{t}L_k(\lambda_k)-r\mid \lambda_k\in L^0(\mathbb{R}^{d\times d}_+,\mathcal{F}_k),\ k=s,\dots,t,\ r\in L^0(\mathbb{R}^d_+,\mathcal{F}_t)\right\}
		\end{align}
		for all $0\leq s\leq t\leq T$. 
	\end{Lemma}
	\begin{proof} For each $\lambda_t\in L^0(\mathbb{R}^{d\times d}_+,\mathcal{F}_t)$ and $r_t\in L^0(\mathbb{R}^d_+,\mathcal{F}_t)$ the random vector $L_t(\lambda_t)-r_t$ is an element of $L^0(-K_t,\mathcal{F}_t)$. Hence, we only have to show that for each $v\in L^0(-K_t,\mathcal{F}_t)$, we can find $\lambda_t\in L^0(\mathbb{R}^{d\times d}_+,\mathcal{F}_t)$ and $r_t\in L^0(\mathbb{R}^d_+,\mathcal{F}_t)$ such that $v=L_t(\lambda_t)-r_t$ a.s.
		
For this, let $v\in L^0(-K_t,\mathcal{F}_t)$, i.e., $v(\omega)\in -K_t(\omega)$ for each $\omega\in\Omega\setminus N$, where $N\in\mathcal{F}_t$ is a set of measure zero. Then $\widetilde{v}:=\mathbbm{1}_{\Omega\setminus N} v\in L^0(-K_t,\mathcal{F}_t)$ satisfies $\widetilde{v}=v$ a.s. and $\widetilde{v}(\omega)\in -K_t(\omega)$ for all $\omega\in \Omega$. Next, we define the set-valued mapping $\omega\mapsto P(\omega)\subseteq\mathbb{R}^{d\times d}\times \mathbb{R}^d$ by
		\begin{align*}
			 P(\omega):=\left\{\left(\lambda,r\right)\in\mathbb{R}^{d\times d}\times \mathbb{R}^d\mid \lambda,r \geq 0, \ \sum_{1\leq i,j\leq d} \lambda^{ij}\left(e^j-\pi^{ij}_t(\omega)e^i\right)-r=\widetilde{v}(\omega)\right\}.
		\end{align*}
		We have $P(\omega)\neq \emptyset$ for each $\omega\in\Omega$ by virtue of $\widetilde{v}(\omega)\in -K_t(\omega)$ for each $\omega\in \Omega$. In addition, 
the mapping $\omega\mapsto \sum_{1\leq i,j\leq d} \lambda^{ij}\left(e^j-\pi^{ij}_t(\omega)e^i\right)-r$ is $\mathcal{F}_t$-measurable for each $(\lambda,r)\in\mathbb{R}^{d\times d}_+\times \mathbb{R}^d_+$, and the mapping $(\lambda,r)\mapsto \sum_{1\leq i,j\leq d} \lambda^{ij}\left(e^j-\pi^{ij}_t(\omega)e^i\right)-r$
is continuous for each $\omega\in\Omega$. Hence, we may apply Theorem 14.36 in \cite{rockafellar2009variational} to find $\lambda_t\in L^0(\mathbb{R}^{d\times d}_+,\mathcal{F}_t)$ and $r_t\in L^0(\mathbb{R}^d_+,\mathcal{F}_t)$ such that $(\lambda_t(\omega),r_t(\omega))\in P(\omega)$ for all $\omega\in\Omega$. This yields $\widetilde{v}(\omega)=\sum_{1\leq i,j\leq d} \lambda^{ij}_t(\omega)\left(e^j-\pi^{ij}_t(\omega)e^i\right)-r(\omega)$ for each $\omega\in \Omega$ and, consequently, we have $v=L_t(\lambda_t)-r_t$ a.s. At last, \eqref{eq:Representation2} follows directly from \eqref{eq:Representation1}.
	\end{proof}

\begin{definition}
		 For any $t\in\{0,\dots,T-1\}$, we define the (convex) cone of reversible orders at time $t$ by 
		 \begin{align*}
		 	\mathcal{R}_t:=\{\lambda\in L^0(\mathbb{R}^{d\times d}_+,\mathcal{F}_t)\mid -L_t(\lambda)\in \mathcal{A}_{t+1}^T\}.
		 \end{align*}
	\end{definition}
The following lemma establishes a suitable decomposition  of the elements of $L^0(\mathbb{R}^{d\times d}_+,\mathcal{F}_t)$ into reversible and
``purely non-reversible'' orders. For the decomposition, one needs that $\mathcal{R}_t$ is closed in probability. To achieve this, the lemma {\em assumes} that $\mathcal{A}_{t+1}^T$ is closed  in probability, a property that is not yet shown at this place. 
 \begin{Lemma}\label{Lemma:ProjectionProperties}
	 	Let $t\in \{0,\dots,T-1\}$ and assume that $\mathcal{A}_{t+1}^T$ is closed in probability. Then for any $\lambda\in L^0(\mathbb{R}^{d\times d}_+,\mathcal{F}_t)$ there is a unique pair (up to null sets) $\lambda_1\in \mathcal{R}_t$ and $\lambda_2\in L^0(\mathbb{R}^{d\times d}_+,\mathcal{F}_t)$ with $\lambda=\lambda_1+\lambda_2$ such that for any decomposition $\lambda=\widetilde{\lambda}_1+\widetilde{\lambda}_2$ with $\widetilde{\lambda}_1\in \mathcal{R}_t$, $\widetilde{\lambda}_2\in L^0(\mathbb{R}^{d\times d}_+,\mathcal{F}_t)$, we have 
	 	\begin{align}\label{eq:ProjectionOptimality}
		 	\Vert \lambda_2\Vert_2\leq \Vert \widetilde{\lambda}_2\Vert_2 \quad
		 	\mathbb{P}\text{-a.s.},
	 	\end{align}
	 	where the inequality is strict on $\{\lambda_2\neq \widetilde{\lambda}_2\}$\ $\mathbb{P}\text{-a.s.}$ and  $\Vert \cdot\Vert_2$ denotes the Euclidean norm on $\mathbb{R}^{d\times d}$.
	 	In addition, the mappings $p_t:L^0(\mathbb{R}^{d\times d}_+,\mathcal{F}_t)\to \mathcal{R}_t$ and $q_t:L^0(\mathbb{R}^{d\times d}_+,\mathcal{F}_t)\to L^0(\mathbb{R}^{d\times d}_+,\mathcal{F}_t)$ defined by $p_t(\lambda)=\lambda_1$ and $q_t(\lambda)=\lambda_2$ have the following properties:
	 	\begin{enumerate}[(i)]
	 		\item \label{item:ProjectionProperty1} For all $\lambda\in L^0(\mathbb{R}^{d\times d}_+,\mathcal{F}_t)$ and all non-negative $\mathcal{F}_t$-measurable scalars $\mu$ we have $p_t(\mu \lambda)=\mu p_t(\lambda)$,
	 		\item \label{item:ProjectionProperty2}  $\mathrm{Image}(q_t)=\{\lambda\in L^0(\mathbb{R}^{d\times d}_+,\mathcal{F}_t)\mid q_t(\lambda)=\lambda\}$,
	 		\item \label{item:ProjectionProperty3} $\mathrm{Image}(p_t)\cap\mathrm{Image}(q_t)=\{0\}$.
\end{enumerate}
	\end{Lemma}  
We refer to  $p_t(\lambda)$ and $q_t(\lambda)$ as the reversible and the 
purely non-reversible part of the order~$\lambda\in L^0(\mathbb{R}^{d\times d}_+,\mathcal{F}_t)$, respectively. The following continuity of the decomposition is the last ingredient for the proof of 
Theorem~\ref{theo:MainResult1}.
%
\begin{Lemma} \label{Lemma:ClosedImQ}
Let $t\in \{0,\dots,T-1\}$ and assume that $\mathcal{A}_{t+1}^T$ is closed in probability. Let $(\lambda_n)_{n\in\mathbb{N}}\subseteq L^0(\mathbb{R}^{d\times d}_+,\mathcal{F}_t)$ converge  $\mathbb{P}\text{-a.s.}$ to some $\lambda\in L^0(\mathbb{R}^{d\times d}_+,\mathcal{F}_t)$. 
Then, $p_t(\lambda_n)\to p_t(\lambda)$ and $q_t(\lambda_n)\to q_t(\lambda)$  $\mathbb{P}\text{-a.s.}$ for $n\to\infty$. Especially, 
$\mathrm{Image}(q_t)$ is closed  in probability.
\end{Lemma}
We postpone the proofs of the two lemmas to make some comments on their use. By the prospective strict no-arbitrage~(NA\textsuperscript{ps}) property, reversible orders can be postponed to later periods $s\in\{t+1,\dots,T\}$. Thus, any order at time~$t$ can be replaced by its purely non-reversible part at time~$t$. On the other hand, if a sequence in
$\mathcal{A}_t^T$ converges, an explosion of the purely non-reversible orders at time~$t$ can be led to a contradiction. The mapping $p_t$ plays the role of the projection of an arbitrary self-financing strategy onto the set of null-strategies in  \cite{schachermayer2004fundamental}.
There, the null-strategies form a linear subspace, which implies that the orthogonal part is automatically self-financing. This property is not available here, and thus we cannot argue with a projection, but with a more complicated decomposition. 

Alternatively, the decomposition in 
Lemma~\ref{Lemma:ProjectionProperties} could also be defined 
on the level of portfolio changes~$\vartheta_t - \vartheta_{t-1}\in L^0(\mathbb{R}^d,\mathcal{F}_t)$. But, in the proof of Lemma~\ref{Lemma:ClosedImQ}, we have to argue 
directly with the orders~$\lambda$.\\

For the convenience of the reader, we recall a lemma on the existence of a measurable subsequence that is applied several times in the following proofs (see, e.g. \cite{schachermayer2004fundamental} and \cite{kabanov2001teacher}).
			\begin{Lemma}[Lemma A.2 of \cite{schachermayer2004fundamental}]\label{Lemma: MeasurableSubsequence}
Let $t\in\{0,\ldots,T\}$. For a sequence $(f_n)_{n\in\mathbb{N}}\subseteq 
L^0(\mathbb{R}^{d\times d}_+, \mathcal{F}_t)$, there is a random subsequence $(\tau_k)_{k\in\mathbb{N}}$, i.e., 
a strictly increasing sequence of $\mathbb{N}$-valued $\mathcal{F}_t$-measurable random variables such that the sequence of random variables $(g_k)_{k\in\mathbb{N}}$ given by $g_k(\omega):=f_{\tau_k(\omega)}(\omega)$, $k\in\mathbb{N}$, converges a.s. in the one-point-compactification $\mathbb{R}^{d\times d}_+\cup\{\infty\}$ to a random variable in $f\in L^0(\mathbb{R}^{d\times d}_+\cup\{\infty\}, \mathcal{F}_t)$. In fact, we may find the subsequence such that
				\begin{align*}
					\Vert f\Vert_{2}=\limsup_{n\to\infty}\Vert f_n\Vert_{2},\ \mathbb{P}\text{-a.s.}
				\end{align*}
				where $\Vert \infty\Vert_{2}=\infty$.
			\end{Lemma}

\begin{proof}[Proof of Lemma~\ref{Lemma:ProjectionProperties}] 
First, we show the existence and uniqueness of the decomposition satisfying (\ref{eq:ProjectionOptimality}).
Fix $\lambda\in L^0(\mathbb{R}^{d\times d}_+,\mathcal{F}_t)$ and define the non-empty set  
		\begin{align*}
			X_\lambda:=\{\widetilde{\lambda}\in \mathcal{R}_t\mid \lambda-\widetilde{\lambda}\in L^0(\mathbb{R}^{d\times d}_+,\mathcal{F}_t)\},
		\end{align*} 
which consists of the first components of the possible decompositions of $\lambda$. 
Under the assumptions made, the convex cone $\mathcal{R}_t$ is closed  in probability and closed under multiplication with non-negative $\mathcal{F}_t$-measurable scalars.
This implies that $X_\lambda$ is closed  in probability and closed under measurable convex combinations. We have to show that
\begin{align*}
			x:=\essinf_{\widetilde{\lambda}\in X_\lambda}\Vert \lambda-\widetilde{\lambda}\Vert_2 
		\end{align*}
is attained and the minimizer is unique. Thus, notice that the set of random variables $\{\Vert \lambda-\widetilde{\lambda}\Vert_2\mid \widetilde{\lambda}\in X_\lambda\}$ is downward directed. Indeed,
for each $\lambda_1,\lambda_2\in X_\lambda$, one has
$\Vert \lambda-\lambda_3\Vert_2 = \Vert \lambda-\lambda_1\Vert_2\wedge\Vert \lambda-\lambda_2\Vert_2$, where  
		\begin{align*}
		X_\lambda\ni\lambda_3:=\mathbbm{1}_{\{\Vert \lambda-\lambda_1\Vert_2\leq \Vert \lambda-\lambda_2\Vert_2\}}\lambda_1+\mathbbm{1}_{\{\Vert \lambda-\lambda_1\Vert_2> \Vert \lambda-\lambda_2\Vert_2\}}\lambda_2.
		\end{align*}
Hence, there is a sequence of random variables $(\lambda_n)_n\subseteq X_\lambda$ such that $\Vert \lambda-\lambda_n\Vert_2\to x$\ $\mathbb{P}\text{-a.s.}$ for $n\to\infty$. From the parallelogram law 
(see, e.g., Lemma 6.51 in \cite{aliprantis2006infinite}) 
of the Euclidean norm on $\mathbb{R}^{d\times d}$ and the convexity of $X_\lambda$, we obtain
		\begin{align}\label{8.10.2018.1}
		\Vert \lambda_n-\lambda_m\Vert_2^2 & = 2\Vert \lambda-\lambda_n\Vert_2^2+2\Vert \lambda-\lambda_m\Vert_2^2-4 \Vert \lambda-\frac{\lambda_n+\lambda_m}{2}\Vert_2^2\nonumber\\
		&\leq 2\Vert \lambda-\lambda_n\Vert_2^2+2\Vert \lambda-\lambda_m\Vert_2^2-4x\quad \mathbb{P}\text{-a.s.}
		\end{align}
(\ref{8.10.2018.1}) implies that $(\lambda_n)_{n\in\mathbb{N}}$ converges  $\mathbb{P}\text{-a.s.}$ to some element of $L^0(\mathbb{R}^{d\times d}_+,\mathcal{F}_t)$. By the closedness of $X_\lambda$, one derives the existence. Uniqueness in the postulated sense also follows from the estimate (\ref{8.10.2018.1}).

This means that the mappings $p_t$ and $q_t$ are well-defined and it remains to show that they satisfy the properties.\\

Ad \eqref{item:ProjectionProperty1}: Let $\mu\geq 0$ be a $\mathcal{F}_t$-measurable random variable. As a consequence of $\mathcal{R}_t$ and $L^0(\mathbb{R}^{d\times d}_+,\mathcal{F}_t)$ being closed under multiplication with non-negative $\mathcal{F}_t$-measurable random variables, we have $X_{\mu\lambda}= \{\mu\widetilde{\lambda}\ |\ \widetilde{\lambda}\in X_\lambda\}$. Then, the assertion follows from the construction of $p_t$ from above.\\

Ad  \eqref{item:ProjectionProperty2}: Let $\lambda\in L^0(\mathbb{R}^{d\times d}_+,\mathcal{F}_t)$. We have $p_t(\lambda)+p_t(q_t(\lambda))\in \mathcal{R}_t+\mathcal{R}_t\subseteq \mathcal{R}_t$ and $\lambda-\left(p_t(\lambda)+p_t(q_t(\lambda))\right)=q_t(\lambda)-p_t(q_t(\lambda))\in L^0(\mathbb{R}^{d\times d}_+,\mathcal{F}_t)$ by definition of $p_t$, thus in particular 
\begin{align}\label{9.10.2018.1}
p_t(\lambda)+p_t(q_t(\lambda))\in X_\lambda.
\end{align}
On the other hand, one has
		\begin{align} \label{eq:ProofProjectionProperties}
	\Vert \lambda-(p_t(\lambda)+p_t(q_t(\lambda))\Vert_2=\Vert q_t(\lambda)-p_t(q_t(\lambda))\Vert_2
	\leq\Vert q_t(\lambda)\Vert_2
	=\Vert \lambda-p_t(\lambda)\Vert_2\quad \mathbb{P}\text{-a.s.},
		\end{align}
where the inequality holds since $p_t(q_t(\lambda))$ is the optimal reversible part of $q_t(\lambda)$.  By 
(\ref{eq:ProofProjectionProperties}), (\ref{9.10.2018.1}), and the uniqueness of the
optimal reversible part in the decomposition of $\lambda$, it follows that $p_t(\lambda)+p_t(q_t(\lambda)) =p_t(\lambda)$\ $\mathbb{P}\text{-a.s.}$ and thus
\begin{align}\label{9.10.2018.2}
\mathrm{Image}(p_t\circ q_t)=\{0\}.
\end{align}
The assertion immediately follows from (\ref{9.10.2018.2}).\\

Ad (iii): Follows immediately from (ii).
\end{proof}
	
\begin{proof}[Proof of Lemma~\ref{Lemma:ClosedImQ}]
We have to show that
\begin{align}\label{10.10.2018.1}
\lambda_n\to \lambda\quad \mathbb{P}\text{-a.s.}\quad\implies\quad
p_t(\lambda_n)\to p_t(\lambda)\quad \mathbb{P}\text{-a.s.}
\end{align}
The property that $\mathrm{Image}(q_t)$ is closed in probability
immediately follows  from Lemma~\ref{Lemma:ProjectionProperties} (ii) 
and (\ref{10.10.2018.1}) by passing to an almost surely converging subsequence. To show (\ref{10.10.2018.1}), we define
for each $n\in\bbn$ the $\mathcal{F}_t$-measurable real-valued random variable 
\begin{align}\label{mu}
\mu_n(\omega):= 1\wedge \inf_{\substack{1\leq i,j\leq d\\ {p_t(\lambda)}^{ij}(\omega)>0}}\frac{\lambda_n^{ij}(\omega)}{{p_t(\lambda)}^{ij}(\omega)}.
\end{align}
One has that $\mu_n p_t(\lambda)\in\mathcal{R}_t$ and $\lambda_n-\mu_n p_t(\lambda)\in L^0(\mathbb{R}^{d\times d}_+,\mathcal{F}_t)$, 
i.e.,  $\mu_n p_t(\lambda)\in X_{\lambda_n}$. This means that we compress the 
transfer matrix~$p_t(\lambda)$ to use it for a (in general not optimal) decomposition of $\lambda_n$ into a reversible and a non-reversible part. Note that in the trivial case that $p_t(\lambda)^{ij}(\omega)=0$ for all $(i,j)$, the compression is irrelevant, here one has 
$\mu_n(\omega)=1$. As $p_t(\lambda_n)$ is the optimal reversible part of $\lambda_n$, it follows that
\begin{align}\label{14.10.2018.1}
\Vert \lambda_n-p_t(\lambda_n)\Vert_2
\le \Vert \lambda_n-\mu_n p_t(\lambda)\Vert_2\quad \mathbb{P}\text{-a.s.}\quad \forall n\in\bbn.
\end{align}
In addition, by 
$\lambda^{ij}\ge p_t(\lambda)^{ij}\ge 0$ for all $i,j=1,\ldots,d$
and $\lambda_n\to\lambda$, we have that $\mu_n\to 1$\ $\mathbb{P}\text{-a.s.}$ Combining this with the triangle inequality of the Euclidean norm, we arrive at 
\begin{align*}
\limsup_{n\to\infty}\Vert \lambda-p_t(\lambda_n)\Vert_2 = 
\limsup_{n\to\infty}\Vert \lambda_n-p_t(\lambda_n)\Vert_2
\le
\limsup_{n\to\infty}\Vert \lambda_n-\mu_n p_t(\lambda)\Vert_2 = 
\Vert \lambda-p_t(\lambda)\Vert_2 \ \mathbb{P}\text{-a.s.},
\end{align*}
where the inequality follows from (\ref{14.10.2018.1}).
Since $p_t(\lambda)$ is the optimal reversible part of $\lambda$, this just means that
\beam\label{11.10.2018.2}
 \Vert \lambda-p_t(\lambda_n)\Vert_2 \to \Vert \lambda-p_t(\lambda)\Vert_2,\quad n\to \infty,\quad \mathbb{P}\text{-a.s.}
\eeam
To complete the proof, we define the  $\mathcal{F}_t$-measurable random variable $\eps(\omega):=\sup\{1/k\ |\ k\in\bbn,\ 
 \Vert p_t(\lambda_n)-p_t(\lambda)\Vert_2(\omega)\ge 1/k\ \mbox{for infinitely many}\ n\}$ that is  strictly positive on the set~$A:=\{p_t(\lambda_n)\not\to p_t(\lambda)\}\in\mathcal{F}_t$.
 Then, we construct the random subsequence~$(\tau_k)_{k\in\bbn}$ recursively by $\tau_0:=0$ and $\tau_k:= \inf\{n\in\bbn\ |\ n>\tau_{k-1},\  \Vert p_t(\lambda_n)-p_t(\lambda)\Vert_2\ge \eps\}$ on $A$ and $\tau_k:=k$ on $\Omega\setminus A$. By construction, we have that  
\beam\label{11.10.2018.1}
P\left( \Vert p_t(\lambda_{\tau_k})-p_t(\lambda)\Vert_2\ge \eps,\ \forall k\in\bbn\ |\ A\right)=1.
\eeam
By $\lambda_n\to\lambda$ and $0\le p_t(\lambda_n)^{ij}\le \lambda^{ij}_n$, one has $\sup_{n\in\bbn}\Vert p_t(\lambda_n)\Vert_2
\le \sup_{n\in\bbn}\Vert \lambda_n\Vert_2<\infty\ \mathbb{P}\text{-a.s.}$
Thus, by Lemma~\ref{Lemma: MeasurableSubsequence}, 
there exists a random subsequence~$(\wt{\tau}_k)_{k\in\bbn}$ of $(\tau_k)_{k\in\bbn}$ and an $f\in L^0(\mathbb{R}^{d\times d}_+,\mathcal{F}_t)$ s.t.  $p_t(\lambda_{\wt{\tau}_k})\to f\ \mathbb{P}\text{-a.s.}$
Together with (\ref{11.10.2018.2}), this implies that  $\Vert \lambda-f\Vert_2 = \Vert \lambda-p_t(\lambda)\Vert_2\ \mathbb{P}\text{-a.s.}$. In addition, we have $f\in X_\lambda$.
On the other hand, by (\ref{11.10.2018.1}), $f\not= p_t(\lambda)$ on $A\ \mathbb{P}\text{-a.s.}$ Since $p_t(\lambda)$ is the unique optimal reversible part of $\lambda$ in the sense of Lemma~\ref{Lemma:ProjectionProperties}, these two properties can only hold simultaneously if $P(A)=0$ and we are done.
\end{proof} 
\begin{Remark}
We note that for the proof of Theorem~\ref{theo:MainResult1}, we only need the weaker assertion that $\mathrm{Image}(q_t)$ is closed in probability. To show this assertion, 
one can restrict oneself to sequences with $\lambda_n=q_t(\lambda_n)$, i.e., $p_t(\lambda_n) = 0$, for all $n\in\bbn$, and the above proof would already be completed with (\ref{11.10.2018.2}).
\end{Remark}
	We are now in the position to prove Theorem~\ref{theo:MainResult1}. 
	As in Kabanov-R\'{a}sonyi-Stricker~\cite{kabanov2003closedness}	
we argue by induction on the periods. The key difference is that reversible orders are postponed
to later periods, instead of being executed and compensated in the same period.	The later is not possible  since the null-strategies do not form a linear space.
	\begin{proof}[Proof of Theorem \ref{theo:MainResult1}] 
		Assume that the bid-ask-process $(\Pi_t)_{t=0}^T$ satisfies (NA\textsuperscript{ps}). Let us prove by a backward induction on $t=T,T-1,\ldots,0$ that $\mathcal{A}_t^T$ is closed in probability. The induction basis $t=T$ is trivial since $\mathcal{A}_T^T$ coincides with $L^0(-K_T,\mathcal{F}_T)$ which is closed in probability. 

Induction step $t+1\leadsto t$: We assume that $\mathcal{A}_{t+1}^T$ is closed in probability for some $t\leq T-1$ and have to show that $\mathcal{A}_t^T$ is closed too. Therefore, let $(\xi_n)_{n\in\mathbb{N}}$ be a sequence in $\mathcal{A}_t^T$ which converges to some $\xi\in L^0(\mathbb{R}^d)$ in probability. Obviously, we may assume that $\xi_n\to\xi$ almost surely by passing to a subsequence. We have to show that $\xi\in \mathcal{A}_t^T$. 

\textit{Step 1.} According to Lemma~\ref{Lemma:Representations}, we may write
		\begin{align}\label{13.10.2018.1}
			\xi_n=\sum_{s=t}^{T}L_s(\lambda_s^n)-r^n, \quad n\in\mathbb{N},
		\end{align}
		where $(\lambda_s^n)_{n\in\mathbb{N}}\subseteq L^0(\mathbb{R}^{d\times d}_+,\mathcal{F}_s)$ for each $s=t,\dots,T$ and $(r^n)_{n\in\mathbb{N}}\subseteq L^0(\mathbb{R}^d_+)$. 
Under the induction hypothesis that $\mathcal{A}_{t+1}^T$ is closed in probability, we apply Lemma \ref{Lemma:ProjectionProperties} in order to decompose $\lambda_t^n$ into $p_t(\lambda_t^n)+q_t(\lambda_t^n)$ and thus
\beao
L_t(\lambda_t^n)=L_t(p_t(\lambda_t^n))+L_t(q_t(\lambda_t^n)),
\eeao
where $p_t(\lambda_t^n)$ is reversible and $q_t(\lambda_t^n)$ is purely non-reversible. This means that $L_t(p_t(\lambda^n_t))\in \mathcal{A}_0^{t}\cap(-\mathcal{A}_{t+1}^T)$. The prospective strict no-arbitrage (NA\textsuperscript{ps}) property implies that 
$\mathcal{A}_0^t\cap\left(-\mathcal{A}_{t+1}^T\right)\subseteq \mathcal{A}_0^{t+1}\cap\left(-\mathcal{A}_{t+1}^T\right)\subseteq \mathcal{A}_{t+1}^T$ and thus $L_t(p_t(\lambda_t^n))\in \mathcal{A}_{t+1}^T$. This allows us to rewrite (\ref{13.10.2018.1}) as
\beao
\xi_n= L_t(q_t(\lambda_t^n))+L_t(p_t(\lambda_t^n))
+\sum_{s=t+1}^{T}L_s(\lambda_s^n)-r^n
=: L_t(q_t(\lambda_t^n)) +x_n
\eeao
with $x_n\in\mathcal{A}_{t+1}^T$. Hence, from now on we can assume w.l.o.g. that $(\lambda_t^n)_{n\in\mathbb{N}}\subseteq \mathrm{Image}(q_t)$.

		\textit{Step 2.} Our next goal is to show that 
		\begin{align}\label{eq:ProofMainTheoClaim}
P(A)=0,\quad\mbox{where}\quad A:=\{\limsup_{n\to\infty}\Vert\lambda^n_t\Vert_2=\infty\}.
		\end{align} 
By Lemma~\ref{Lemma: MeasurableSubsequence}, we may pass to a measurable subsequence $(\tau_k)_{k\in\mathbb{N}}$ such that for a.e. $\omega\in A$ we have $\lambda^{\tau_k(\omega)}_t(\omega)\neq 0$ for all $k\in\bbn$ and $\lim_{k\to\infty}\Vert \lambda^{\tau_k(\omega)}_t(\omega)\Vert_2=\infty$. Then, by the stability of $\mathrm{Image}(q_t)$ under multiplication with non-negative $\mathcal{F}_t$-measurable scalars (see Lemma~\ref{Lemma:ProjectionProperties}(i)), we find that $\widetilde{\lambda}_t^n:=\frac{\lambda_t^{\tau^{n}}}{\Vert \lambda_t^{\tau_n}\Vert_2}\mathbbm{1}_A$ belongs to $\mathrm{Image}(q_t)$  and, in addition, we define
		\begin{align*}
			\widetilde{\lambda}_s^n:=\frac{\lambda_s^{\tau_n}}{\Vert \lambda_t^{\tau_n}\Vert_2}\mathbbm{1}_A\in L^0(\mathbb{R}^{d\times d}_+,\mathcal{F}_s)\quad\text{for}\ s=t+1,\dots,T\quad \text{and}\quad \widetilde{r}^n:=\frac{r^n}{\Vert \lambda_t^{\tau_n}\Vert_2}\mathbbm{1}_A\in L^0(\mathbb{R}^d_+).
		\end{align*}  
We have $\sum_{s=t}^{T}L_s(\widetilde{\lambda}^n_s)-\widetilde{r}^n = \mathbbm{1}_A \xi_{\tau_n}/\Vert \lambda_t^{\tau_n}\Vert_2\to 0$ a.s. Now, we may apply once again Lemma~\ref{Lemma: MeasurableSubsequence} to find a measurable subsequence $(\sigma_k)_{k\in\mathbb{N}}$ such that
		\begin{align}
			\widetilde{\lambda}_t:=\lim\limits_{k\to\infty}\widetilde{\lambda}^{\sigma_k}_t
		\end{align}
		exists and $\Vert \widetilde{\lambda}_t\Vert_2=\lim\limits_{k\to\infty}\Vert\widetilde{\lambda}^{\sigma_k}_t\Vert_2$.
		Consequently, $L_t(\widetilde{\lambda}^{\sigma_k}_t)\to L_t(\widetilde{\lambda}_t)$ and thus the sequence 
\beao
\left(\sum_{s=t+1}^{T}L_s(\widetilde{\lambda}^{\sigma_k}_s)-\widetilde{r}^{\sigma_k}\right)_{k\in\mathbb{N}}\subseteq \mathcal{A}_{t+1}^T
\eeao 
converges to $-L_t(\widetilde{\lambda}_t)$. Since $\mathcal{A}_{t+1}^T$
is closed and due to Lemma~\ref{Lemma:Representations}, the limit can be written 
as $\sum_{s=t+1}^{T}L_s(\widetilde{\lambda}_s)-\widetilde{r}$, i.e., we have

		\begin{align*}
			L_t(\widetilde{\lambda}_t)+\sum_{s=t+1}^{T}L_s(\widetilde{\lambda}_s)-\widetilde{r}=0\quad \mathbb{P}\text{-a.s.}
		\end{align*}  
with $\widetilde{\lambda}_s\in 	 L^0(\mathbb{R}^{d\times d}_+,\mathcal{F}_s)$ and 	
$\widetilde{r}\in L^0(\mathbb{R}^d_+,\mathcal{F}_s)$. Thus we have 
that $\widetilde{\lambda}_t$ is reversible, i.e., $\widetilde{\lambda}_t\in \mathcal{R}_t=\mathrm{Image}(p_t)$. However, on the other hand, the sequence $(\widetilde{\lambda}^{\sigma_k}_t)_{k\in\mathbb{N}}$ belonged to $\mathrm{Image}(q_t)$, thus, by Lemma \ref{Lemma:ClosedImQ}, $\widetilde{\lambda}_t\in\mathrm{Image}(q_t)$. Therefore $\widetilde{\lambda}_t\in \mathrm{Image}(p_t)\cap \mathrm{Image}(q_t)$, hence $\widetilde{\lambda}_t=0$ a.s. according to 
		Lemma~\ref{Lemma:ProjectionProperties} \eqref{item:ProjectionProperty3}. 
		Since $\mathbb{P}(A)=\mathbb{P}(\widetilde{\lambda}_t\neq 0)$,
this is only possible if $P(A)=0$, i.e., \eqref{eq:ProofMainTheoClaim} holds true.

		\textit{Step 3.} According to step 2, we can apply Lemma \ref{Lemma: MeasurableSubsequence} to find a measurable subsequence $(\tau_k)_{k\in\mathbb{N}}$ such that $\lambda^{\tau_k}_t\to \lambda_t \in L^0(\mathbb{R}^{d\times d}_+,\mathcal{F}_t)$\ $\mathbb{P}$-a.s. for $k\to\infty$ and, consequently, $L_t(\lambda^{\tau_k}_t)\to L_t(\lambda_t)$ a.s. Hence, $\sum_{s=t+1}^{T}L_s(\lambda_s^{\tau_k})-r^{\tau_k}$ converges a.s. to $\xi-L_t(\lambda_t)$, which, by the induction hypothesis, belongs to $\mathcal{A}_{t+1}^T$. This implies that $\xi\in \mathcal{A}_t^T$. 
	\end{proof}

	Finally, we finish up the remaining proofs. Notice that every result is a standard consequence of the set $\mathcal{A}_0^T$ being closed in probability under the (NA\textsuperscript{ps}) condition, hence we only give the respective references and point out where some changes are needed. 
	\begin{proof}[Proof of Corollary \ref{cor:MainResult1}]
	It suffices to repeat the arguments on page 29 between lines 5-33 of the proof of Theorem 2.1 in \cite{schachermayer2004fundamental} with $\mathcal{A}_0^T$ (instead of $\widetilde{\mathcal{A}}_T$), which is closed by Theorem~\ref{theo:MainResult1}. 
	\end{proof}
	\begin{proof}[Proof of Theorem \ref{theo:MainResult2}] 
	
(NA\textsuperscript{wps})$\Rightarrow \exists$ CPS: According to the (NA\textsuperscript{wps}) condition there is a bid-ask process $(\widetilde{\Pi}_t)_{t=0}^T$ with $\widetilde{\Pi}_t\leq \Pi_t$  a.s. for all $t=0,\dots,T$ satisfying (NA\textsuperscript{ps}). Corollary \ref{cor:MainResult1} implies that $(\widetilde{\Pi}_t)_{t=0}^T$ admits a CPS, which is obviously a CPS for $(\Pi_t)_{t=0}^T$ as well.
	
$\exists$ CPS $\Rightarrow$ (NA\textsuperscript{wps}):
		It is again sufficient to repeat the arguments on page 30 between lines 1-12 of the proof of Theorem 2.1 in \cite{schachermayer2004fundamental} to define a frictionless bid-ask process $(\widetilde{\Pi}_t)_{t=0}^T$, i.e., $\wt{\pi}_t^{ij}=1/\wt{\pi}_t^{ji}$, with $\widetilde{\Pi}_t\leq \Pi_t$ a.s. for all $t=0,\dots,T$ satisfying (NA), which in the frictionless case coincides with (NA\textsuperscript{ps}) by Proposition~\ref{prop:Comparision}. Thus $(\Pi_t)_{t=0}^T$ satisfies (NA\textsuperscript{wps}).
	\end{proof}
	\begin{proof}[Proof of Corollary \ref{cor:MainResult2}]
		This is a well known consequence of the existence of a consistent price system. We may, for example, use Proposition 3.2.6 in \cite{kabanov2009markets} to see that the existence of a consistent price system implies $\overline{\mathcal{A}_0^T}\cap L^0(K_T,\mathcal{F}_T)\subseteq L^0(\partial K_T,\mathcal{F}_T)$. To complete the 
proof, we observe that  $\partial K_T\cap\mathbb{R}^d_+=\{0\}$. Indeed, by 
$\pi^{ij}<\infty$, the existence of a $v\in \mathbb{R}^d_+\setminus \{0\}$ and a sequence $(v_n)_{n\in\bbn}\subseteq \bbr^d\setminus  K_T(\omega))$ with $v_n\to v$ can easily be led to a contradiction.
\end{proof}
\begin{Remark}
	Our results can be extended to the Kabanov model as defined in Subsection~3.2 of \cite{kabanov2009markets}, which, in addition to the barter market considered here, also covers a wider range of models, e.g., models of a barter market where a bank account is charged the transaction costs and models where baskets of assets are exchanged. To see this, we briefly highlight the minor adjustments. On the other hand, the proofs of Lemmas~\ref{Lemma:ProjectionProperties} and \ref{Lemma:ClosedImQ} are
	heavily based on the polyhedral structure of the solvency cones. The key argument that $\mu_n$ defined in (\ref{mu}) converges to $1$ does not work for general closed solvency cones. 
	
	The Kabanov model is defined as follows. Let $((X_t^i)_{t=0}^T)_{i\in\bbn}$ be a sequence of adapted $\mathbb{R}^d$-valued processes, such that for all $t$ and $\omega$ the set $\{i\in\mathbb{N}\mid X^i_t(\omega)\neq 0\}$ is non-empty and finite, and set 
	\begin{align*}
	-K_t(\omega):= \mathrm{cone}(X^i_t(\omega)\mid i\in\mathbb{N}).
	\end{align*}
	In this case, $K=(K_t)_{t=0}^T$ defined by $K_t(\omega):=-(-K_t(\omega))$ is called a cone-valued process. In addition, we assume $\mathbb{R}^d_+\setminus\{0\}\subseteq \mathrm{int } K_T(\omega)$ for all $\omega$ (which corresponds to the possibility to freely dispose of assets and $\pi^{ij}_T<\infty$ for all $i,j$ in the base model) and $-K_T(\omega)\cap \bbr^d_+ = \{0\}$ (which corresponds to  $\pi^{ij}_T\le \pi^{ik}_T\pi^{kj}_T$ and $\pi^{ii}_T=1$ for all $i,j,k$). The cone of hedgeable claims attainable from zero endowment by trading between $s$ and $t$ is given by  $\mathcal{A}_s^t=\sum_{k=s}^{t}L^0(-K_k,\mathcal{F}_k),\ s\leq t$. 
	The (NA) and (NA\textsuperscript{ps}) conditions are defined 
	accordingly. We now sketch how the arguments of the previous proofs can be applied in this more general setting. Let $I_t(\omega):=\sup\{n\in\mathbb{N}\mid  X^n_t(\omega)\not=0\}$ for $\omega\in\Omega$ and $t=0,\dots,T$. The assumptions above guarantee that $I_t$ is a $\mathbb{N}$-valued $\mathcal{F}_t$-measurable random variable. 
	By $\Omega=\bigcup_{I\in\mathbb{N}}\{I_t=I\}$, the arguments from 
	Lemmas~\ref{Lemma:Representations}, \ref{Lemma:ProjectionProperties}, and \ref{Lemma:ClosedImQ} can be separately applied on the sets $\{I_t=I\}$ for $I\in\mathbb{N}$. In particular, portfolio changes can be represented by
	$L_t(\lambda_t):= \sum_{i=1}^{I_t}\lambda^i_tX^i_t$, where  $\lambda^i_t\in L^0(\mathbb{R}_+,\mathcal{F}_t)$. The Euclidean norm on $\bbr^{d\times d}$ 
	that is used for the decomposition of an order into the reversible and the purely non-reversible part is replaced by 
	\beao
	\Vert \lambda\Vert_2:=\sqrt{\sum_{i=1}^{I}(\lambda^i)^2}\quad \text{on}\ \{I_t=I\}.
	\eeao
	With these adjustments, Theorem~\ref{theo:MainResult1} extends to the Kabanov model
	by arguing along the lines of the original proofs. We note again that it is crucial that for fixed $\omega$, only linear combinations 
	from finitely many $X^i_t(\omega)$ have to be considered. 
	
	At last, we say that $K=(K_t)_{t=0}^T$ satisfies the (NA\textsuperscript{wps}) property if there is a cone-valued process $\wt{K}=(\wt{K}_t)_{t=0}^T$ with the (NA\textsuperscript{ps}) property such that $K_t(\omega)\subseteq \wt{K}_t(\omega)$ for all $\omega$ and $t$. Then, Theorem~\ref{theo:MainResult2} holds true in the Kabanov model as well. Indeed, (NA\textsuperscript{wps}) implies the existence of a consistent price system for $K=(K_t)_{t=0}^T$ and, on the other hand, given a CPS $Z=(Z_t)_{t=0}^T$, the cone-valued process $\wt{K}=(\wt{K}_t)_{t=0}^T$ defined by $\wt{K}_t(\omega):=(\mathbb{R}_+Z_t(\omega))^\star$ satisfies (NA\textsuperscript{ps}) and $K_t(\omega)\subseteq \wt{K}_t(\omega)$ for all $t$ and $\omega$, i.e., $K$ satisfies (NA\textsuperscript{wps}).\\
	
	In addition, our reasoning to show the closedness of $\mathcal{A}_0^T$ can also be applied to models with incomplete information such
	as those considered in \cite{bouchard2006no,de2007no}, where arguing on the level of orders is quite natural.
\end{Remark}	
\section{(Counter-)Examples}\label{sec:Examples}

We start with two very simple examples that illustrate the difference between (NA\textsuperscript{r}), (NA\textsuperscript{ps}), and (NA\textsuperscript{wps}) and the need to consider CPSs which do not lie in the relative interior of the bid-ask spread. 

	\begin{Example}[(NA\textsuperscript{ps}) $\not\Rightarrow$ (NA\textsuperscript{r})]\label{example:WPros_Pros_Robust} We consider a two-asset one-period model with a bank account that does not pay interest and one stock with bid-price $(\underline{S}_t)_{t=0,1}$ and ask-price $(\overline{S}_t)_{t=0,1}$. The deterministic prices are illustrated in Figure~\ref{fig:NApNotNAr} below.
In this special case, 
a (strictly) consistent price system corresponds to a pair~$(\widetilde{S},\mathbb{Q})$ consisting of a measure $\mathbb{Q}\sim \mathbb{P}$ and a $\mathbb{Q}$-martingale $\widetilde{S}$ taking its values in (the relative interor of) $[\underline{S},\overline{S}]$. 
For more details regarding models with a bank account see, e.g., Section 3 in \cite{rokhlin2008constructive}.
 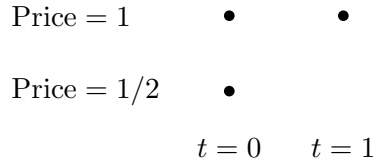
\begin{figure}[H] 
			\centering
			\begin{tikzpicture}
				\fill (1,0) circle (2pt);
				\fill (1,1) circle (2pt);
				\fill (2.5, 1) circle (2pt);
				\node at (1, -0.75) {$t=0$};
				\node at (2.5, -0.75) {$t=1$};
				\node[align=left, right] at (-2,0) {$\mathrm{Price}=1/2$};
				\node[align=left, right] at (-2,1) {$\mathrm{Price}= 1$};
			\end{tikzpicture}
			\caption{Deterministic model satisfying (NA\textsuperscript{ps}) and (NA\textsuperscript{s}), but not (NA\textsuperscript{r}). At $t=0$ the bid-price $\underline{S}_0$ equals $1/2$ and the ask price $\overline{S}_0$ equals $1$; at $t=1$ the market is frictionless with price $\underline{S}_1=\overline{S}_1=1$.}
			\label{fig:NApNotNAr}
		\end{figure}
		Obviously, the market admits the unique consistent price process $\widetilde{S}\equiv 1$. $\wt{S}$ is not a strictly consistent price process since $1\notin (1/2,1)$. Hence, the model cannot satisfy (NA\textsuperscript{r}). Also the Penner-condition, i.e., $L^0(K^0_t,\mathcal{F}_{t-1})\subseteq L^0(K^0_{t-1},\mathcal{F}_{t-1})$ for all $t$, is not satisfied. On the other hand, we have that
$\mathcal{A}_0^0\cap(-\mathcal{A}_1^1) =\mathrm{cone}\left(e^2-e^1\right)\subseteq \mathrm{cone}\left(e^2-e^1, e^1-e^2, -e^1, -e^2\right)
 = \mathcal{A}_1^1$, i.e., only a long stock position built up at time~$0$ can be liquidated without losses at time~$1$, but the purchase of the stock (asset~$2$) can also be postponed to time~$1$. Thus the model satisfies (NA\textsuperscript{ps}).  
\end{Example}

\begin{Example}[(NA\textsuperscript{wps}) $\not\Rightarrow$ (NA\textsuperscript{ps})]\label{example:WPros_Pros} We consider the following variant of Example~\ref{example:WPros_Pros_Robust}:
		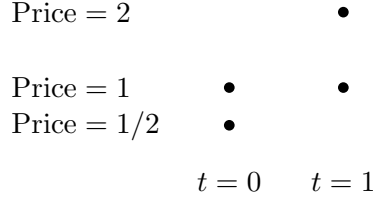
\begin{figure}[H]
			\centering
			\begin{tikzpicture}
				\fill (1,0) circle (2pt);
				\fill (1,0.5) circle (2pt);
				\fill (2.5, 0.5) circle (2pt);
				\fill (2.5, 1.5 ) circle (2pt);
				\node at (1, -0.75) {$t=0$};
				\node at (2.5, -0.75) {$t=1$};
				\node[align=left, right] at (-2,0) {$\mathrm{Price}=1/2$};
				\node[align=left, right] at (-2,0.5) {$\mathrm{Price}= 1$};
				\node[align=left, right] at (-2,1.5) {$\mathrm{Price}= 2$};
			\end{tikzpicture}
			\caption{Deterministic model satisfying (NA\textsuperscript{wps}) but not (NA\textsuperscript{ps}). One has $\underline{S}_0=1/2$, $\overline{S}_0=1$, $\underline{S}_1=1$, and $\overline{S}_1=2$.}
			\label{fig:NAwpNotNAp}
		\end{figure}
		The market still admits the unique consistent price 
		process~$\widetilde{S}\equiv 1$, but now fails (NA\textsuperscript{ps}) 
since $\mathcal{A}_0^0\cap(-\mathcal{A}_1^1) =\mathrm{cone}\left(e^2-e^1\right)\not\subseteq \mathrm{cone}\left(e^2-2 e^1, e^1-e^2, -e^1, -e^2\right)
 = \mathcal{A}_1^1$. On the other hand, the model satisfies (NA\textsuperscript{wps})
since the more favorable bid-ask process in Figure~\ref{fig:NApNotNAr} satisfies (NA\textsuperscript{ps}). 
\end{Example}

	Finally, we provide an example showing that (NA\textsuperscript{wps}) cannot be replaced by the ``next weaker'' condition that there exists a more favorable market, i.e., a bid-ask process $(\widetilde{\Pi}_t)_{t=0}^T$ with $\widetilde{\Pi}_t\leq \Pi_t$ for each $t=0,\dots,T$, such that $(\widetilde{\Pi}_t)_{t=0}^T$ satisfies (NA) and 
	\begin{align}
	\mathcal{A}_0^t\cap\left(-\widetilde{\mathcal{A}}_t^T\right)\subseteq \widetilde{\mathcal{A}}_t^T\quad \mbox{for all } t=0,\dots,T\label{eq:ExampleOtherCondi}
	\end{align}
(cf. Remark~\ref{20.11.2018.2}). We show that there is a bid-ask process~$(\Pi_t)_{t=0}^3$ with four assets satisfying condition \eqref{eq:ExampleOtherCondi} which allows for an approximate arbitrage. Hereby, in the spirit of the basic Example~3.1 of Schachermayer~\cite{schachermayer2004fundamental}, which can be used to achieve an approximate arbitrage, the example is based on the idea of two {\em consecutive} approximate hedges. There exists a more favorable bid-ask process $(\widetilde{\Pi}_t)_{t=0}^3$ s.t. (\ref{eq:ExampleOtherCondi}) holds, but which only turns the first approximate hedge into a perfect hedge and thus the model still satisfies (NA).

The example highlights the importance of a possible ``cascade'' of approximate hedges, which is, to the best of our knowledge, a phenomenon not discussed in the previous literature. It is also of interest for the discussion of adjusted bid-ask processes as introduced in Jacka, Berkaoui, and Warren~\cite{jacka2008no}(see Remark~\ref{26.11.2018.1}). 

\begin{Example}[A cascade of approximate hedges]\label{example:Strengthofthecondition}
		Let $T=3$, $\Omega=\mathbb{N}^2\times \left\{-1/2,1/2\right\}^2$, $\mathcal{F}=2^\Omega$ and all states have positive probability.
In addition, the information structure is given by $\mathcal{F}_0=\{\emptyset,\Omega\}$, \begin{align*}
		\mathcal{F}_1&=\sigma\left(\{\{(n,m,i,j)\mid (m,i,j) \in \mathbb{N}\times\{-1/2, 1/2\}^2\}\mid n\in\mathbb{N}\}\right),\\ \mathcal{F}_2&=\sigma\left( \{\{(n,m,i,-1/2), (n,m,i,1/2)\}\mid (n,m,i)\in \mathbb{N}^2\times \{-1/2,1/2\}\}\right), 
		\end{align*}
		and $\mathcal{F}_3=2^\Omega=\mathcal{F}$. This means $n$ is revealed at time $1$, $m$ and $i$ are revealed at time $2$ and, at last, $j$ is revealed at time $3$. Next, we define a bid-ask process $(\Pi_t)_{t=0}^3$
		depending on parameter $a>0$ for $t=0,1$ as follows
		\begin{align*}
		&\Pi_0=
		\begin{pmatrix}
		1 & 1 & 1 & 1\\
		a & 1 & \cdot & \cdot\\
		a & \cdot & 1& \cdot\\
		a & \cdot & \cdot &1 
		\end{pmatrix},\quad 
		\Pi_1\equiv
		\begin{pmatrix}
		1 & a & a & a\\
		a & 1 & \cdot & \cdot\\
		a & \cdot & 1& \cdot\\
		1 & \cdot& \cdot & 1 
		\end{pmatrix}
		\end{align*}
		and for $t=2,3$ depending on the state $(n,m,i,j)\in\mathbb{N}^2\times\{-1/2,1/2\}^2$ as
		\begin{align*}
		\Pi_2(n,m,i,j)=
		\begin{pmatrix}
		1 & a & a & a\\
		a & 1& \cdot & \cdot\\
		\frac{1}{1+i} & \cdot & 1& \cdot\\
		\frac{1}{1-\frac{i}{n}} & \cdot& \cdot & 1 
		\end{pmatrix},\quad
		\Pi_3(n,m,i,j)=
		\begin{pmatrix}
		1 & a & a & a\\
		\frac{1}{1+\frac{1}{4}+j} & 1 & \cdot & \cdot\\
		\frac{1}{1+i}\frac{1}{1-\frac{j}{m}} & \cdot & 1& \cdot\\
		a & \cdot& \cdot & 1 
		\end{pmatrix}.
		\end{align*}
		The missing entries are specified via the direct transfer over the first asset, i.e., $\pi^{ij}_t:=\pi^{i1}_t\pi^{1j}_t$ for $2\leq i\neq j\leq 4$. This means that the first asset plays the role of a money market account and the assets 2,3, and 4 represent risky stocks. Finally, we choose the parameter~$a$ prohibitively high such that the corresponding transfers are unattractive, more precisely, we set $a:=5>4$.
		 
The market is actually frictionless with special short- and long-selling constraints. Asset~$2$
yields the random return~$1/4+j$, $j\in\{-1/2,1/2\}$.  It can be approximately hedged, yielding an extra profit, by the return of asset~3 between time 2 and time 3, that reads $-j/m$. 
On the other hand, asset~3 has to be bought already at time~0 which leads to the 
prior random return~$i$, $i\in\{-1/2,1/2\}$. The latter return can be approximately hedged by asset~4. This means that there is a cascade of approximate hedges -- hedge asset~2 by asset~3 and asset~3 by asset~4 -- leading to an approximate arbitrage. 
		
{\em Step 1.} Let us show that $(\Pi_t)_{t=0}^3$ allows for an approximate arbitrage, i.e., 
$\overline{\mathcal{A}_0^3}\cap L^0(\mathbb{R}^4_+)\supsetneq\{0\}$. Therefore, we define for fixed $k\in\mathbb{N}$ the following strategy. For $t=0$, we set
		\begin{align*}
		\xi_0^{k}= e^2-e^1+k(e^3-e^1)+k^2(e^4-e^1)\in -K_0.
		\end{align*}
		For $t=1$, we define 
		\begin{align*}
		\xi_1^{k}(n,m,i,j)=\left[k^2-\left(k\wedge n\right)k\right]\left(e^1-e^4\right)\in-K_1(n,m,i,j).
		\end{align*}
		For $t=2$, we define
		\begin{align*}
		\xi_2^{k}(n,m,i,j)&=\left[\left(k\wedge n\right)k\left(1-\frac{i}{n}\right)\right]\left(e^1-\frac{1}{1-\frac{i}{n}}e^4\right)\\
		&+\left[\left(k-\frac{m}{1+i}\wedge k\right)\left(1+i\right)\right]\left(e^1-\frac{1}{1+i}e^3\right)\in -K_2(n,m,i,j).
		\end{align*}
		Finally, at $t=3$ we liquidate the remaining positions in the assets $2$ and $3$. Thus, we define
		\begin{align*}
		\xi_3^{k}(n,m,i,j)&=\left[1+\frac{1}{4}+j\right]\left(e^1-\frac{1}{1+\frac{1}{4}+j}e^2\right)\\ &+ \left[\left(\frac{m}{1+i}\wedge k\right)\left(1+i\right)\left(1-\frac{j}{m}\right)\right]\left(e_1-\frac{1}{1+i}\frac{1}{1-\frac{j}{m}}e^3\right),
		\end{align*}
		which belongs to $-K_3(n,m,i,j)$. Thus $v^{k}=\xi_0^{k}+\xi_1^{k}+\xi_2^{k}+\xi_3^{k}$ belongs to $\mathcal{A}_0^T$ and we have
		\begin{align*}
		v^{k}(n,m,i,j)=\left[\frac{1}{4}+ik\left(1-\frac{n\wedge k}{n}\right)+j\left(1-\frac{1+i}{m}\left(\frac{m}{1+i}\wedge k\right)\right)\right]e^1.
		\end{align*}
		Finally, letting $k\to\infty$, we obtain $v\in \overline{\mathcal{A}_0^T}$ given by
		\begin{align*}
		v(n,m,i,j)=\lim\limits_{k\to\infty}v^{k}(n,m,i,j)=\frac{e^1}{4},
		\end{align*}
		which is the desired asymptotic arbitrage. Hence, the model cannot admit a 
		CPS (see Proposition~3.2.6. in \cite{kabanov2009markets}).

Next, we introduce the bid-ask process $(\widetilde{\Pi}_t)_{t=0}^3$ given by $\widetilde{\Pi}_0=\Pi_0$, $\widetilde{\Pi}_2=\Pi_2$, $\widetilde{\Pi}_3=\Pi_3$, and 
		\begin{align*}
		\widetilde{\Pi}_1\equiv
		\begin{pmatrix}
		1 & a & 1 & 1\\
		a & 1 & \cdot & \cdot\\
		a & \cdot & 1& \cdot\\
		1 & \cdot& \cdot & 1 
		\end{pmatrix},
		\end{align*}  
		which satisfies $\widetilde{\Pi}_t\leq \Pi_t$ for all $t=0,1,2,3$.  We want to show that $(\widetilde{\Pi}_t)_{t=0}^3$ has the (NA) property and satisfies $\mathcal{A}_0^t\cap -\widetilde{\mathcal{A}}_t^T\subseteq \widetilde{\mathcal{A}}_t^T$ for $t=0,1,2,3$.
		
{\em Step 2.} We start with the (NA) property for $(\widetilde{\Pi}_t)_{t=0}^T$. 
Let $\wt{v}\in \widetilde{\mathcal{A}}_0^3$ with $\wt{v}^i=0$ for $i=2,3,4$ and $\wt{v}^1\ge 0$ a.s. We have to show that this already implies $\wt{v}^1=0$ a.s.
We may pass to a $v\in\widetilde{\mathcal{A}}_0^3$ with 
$v^i=0$ for $i=2,3,4$ a.s. and $v^1\ge\wt{v}^1$ s.t. $v$ can be represented solely by transfers 
\beam\label{20.11.2018.1}
\lambda_t^{1j}\ \mbox{with}\ \pi_t^{1j}<a\quad
\mbox{and}\quad \lambda_t^{j1}\ \mbox{with}\ \pi_t^{j1}<a. 
\eeam
Indeed, purchasing an asset $i\in\{2,3,4\}$ at price~$a=5$ or short-selling it at price~$1/a=1/5$ (in terms of the asset~$1$) leads to a sure loss after liquidating this position afterwards. Hence, we only need to consider $v=\xi_0+\xi_1+\xi_2+\xi_3\in \widetilde{\mathcal{A}}_0^3$,
where $\xi_0\in\mathrm{cone}(e^2-e^1)$, $\xi_1\in L^0(\mathrm{cone}(e^3-e^1,e^4-e^1),\mathcal{F}_1)$, $\xi_2\in L^0(\mathrm{cone}(e^1-\widetilde{\pi}^{3,1}_2e^3, e^1-\widetilde{\pi}^{4,1}_2e^4),\mathcal{F}_2)$, and $\xi_3\in L^0(\mathrm{cone}(e^1-\widetilde{\pi}^{2,1}_3 e^2, e^1-\widetilde{\pi}^{3,1}_3 e^3),\mathcal{F}_3)$ with the additional restrictions $\xi^3_1+\xi^3_2\geq 0$ a.s. and $\xi^4_1+\xi^4_2=0$ a.s. Under the assumptions above, we get 
		\begin{align}\label{eq:Ex4.3LiquidationValue}
		0\leq v^1(n,m,i,j)\leq \xi_0^2\cdot(j+1/4)+i\cdot\left(\xi^3_1(n)-\frac{\xi_1^4(n)}{n}\right)+\left(\xi^3_1(n)+\xi^3_2(n,m,i)\right)\frac{-j}{m}(1+i)
		\end{align}
		for each state $(n,m,i,j)\in\Omega$ with $\xi_0^2\geq 0$, $\xi^3_1(n)\geq 0$, $\xi^3_2(n,m,i)\leq 0$ and $\xi^4_1(n)\geq 0$. Hereby the notation highlights the 
required measurability of the random variables. We have 
$\xi^3_1(n)+\xi^3_2(n,m,i)\le \xi_1^3(n)$, i.e., the investment in asset~3 between $t=2$ and $t=3$ is bounded from above by the $\mathcal{F}_1$-measurable random variable $\xi_1^3$. Consequently, the third summand in equation \eqref{eq:Ex4.3LiquidationValue} becomes arbitrarily small for large $m\in\mathbb{N}$. But, this implies that the sum of the first two terms, i.e., 
\beao
\xi_0^2\cdot(j+1/4)+i\cdot\left(\xi^3_1(n)-\frac{\xi_1^4(n)}{n}\right),
\eeao
that does not depend on $m$, has to be almost surely non-negative, which is only possible if \begin{align}\label{eq:Ex4.3HedgingRatio}
		\xi_0^2=0\quad \mathrm{and} \quad
		\xi_1^3(n)=\frac{\xi_1^4(n)}{n}\ \mathrm{for\ each}\ n\in\mathbb{N}.
		\end{align}  
		But then \eqref{eq:Ex4.3LiquidationValue} reduces to 
		\begin{align*}
		0\leq v^1(n,m,i,j)\leq\left(\xi^3_1(n)+\xi^3_2(n,m,i)\right)\frac{-j}{m}(1+i).
		\end{align*}
		Taking $j=1/2$, this implies $\xi^3_1(n)=-\xi_2^3(n,m,i)$ and, consequently, $v^1\equiv 0$. Hence $(\wt{\Pi}_t)_{t=0}^3$  satisfies (NA).

{\em Step 3.} Let us now show $\mathcal{A}_0^t\cap -\widetilde{\mathcal{A}}_t^T\subseteq \widetilde{\mathcal{A}}_t^T$ for $t=1,2,3,4$ (for $t=0$, there is nothing to show). This is akin to the proof in step~2. As in (\ref{20.11.2018.1}), we can restrict to portfolios which can be represented by transfers that do not trade at price~$a=5$. Indeed, since we only consider positions that can be liquidated for sure, the cancellation of a transfer at price~$a$ (with re-transfer at a later time to asset~1), would lead to a strict improvement and thus an arbitrage.
Since this would contradict to the (NA) property of $(\widetilde{\Pi}_t)_{t=0}^3$ shown in Step~2, we can exclude such silly trades in the following considerations.

By the arguments leading to (\ref{eq:Ex4.3HedgingRatio}), it follows that $e^2-e^1\not\in -\widetilde{\mathcal{A}}_1^3$. Consequently, we have $\mathcal{A}_0^1\cap -\widetilde{\mathcal{A}}_1^3=\mathrm{cone}(e^3-e^1,e^4-e^1)+ L^0(\mathrm{cone}(e^1-e^4),\mathcal{F}_1)\subseteq \widetilde{\mathcal{A}}_1^3$. 

Now consider the case $t=2$. Let $w\in \mathcal{A}_0^2\cap - \widetilde{\mathcal{A}}_2^3$, i.e. we may write $w=\xi_0+\xi_1+\xi_2= -\widetilde{\xi}_2-\widetilde{\xi}_3$ for $\xi_0\in \mathrm{cone}(e^2-e^1,e^3-e^1, e^4-e^1)$, $\xi_1\in L^0(\mathrm{cone}(e^4-e^1),\mathcal{F}_1)$, $\xi_2,\widetilde{\xi}_2\in L^0(\mathrm{cone}(e^1-\pi^{31}_2e^3, e^1-\pi^{41}_2e^4),\mathcal{F}_2)$ and $\widetilde{\xi}_3\in L^0(\mathrm{cone}(e^1-\pi^{21}e^2, e^1-\pi^{31}e^3),\mathcal{F}_3)$ with the restrictions $\xi_0^3+\xi_2^3+\widetilde{\xi}_2^3\geq 0$, $\xi_0^4+\xi_1^4\geq 0$ and $\xi_0^4+\xi_1^4+\xi_2^4+\widetilde{\xi}_2^4=0$. Indeed, this is a consequence of (NA) and the avoidance of silly trades. But, then we may consider $v:=w-w=\xi_0+\xi_1+\xi_2-(-\widetilde{\xi}_2-\widetilde{\xi}_3)$. Note that we have $v^i=0$ for $i=2,3,4$, which uniquely determines $v^1$ as 
	\begin{align*}
	 v^1(n,m,i,j)=\xi_0^2\cdot(j+1/4)+i\cdot\left(\xi^3_0-\frac{\xi_0^4+\xi_1^4(n)}{n}\right)+\left(\xi^3_0+\xi^3_2(n,m,i)+\widetilde{\xi}^3_2(n,m,i)\right)\frac{-j}{m}(1+i).
	\end{align*}
	On the other hand, we also have $v^1=0$. Since the first two terms do not depend on 
	$m\in\bbn$, we must have that
	\begin{align*}
\xi_0^2\cdot(j+1/4)+i\cdot\left(\xi^3_0-\frac{\xi_0^4+\xi_1^4(n)}{n}\right) = 0.
	\end{align*}
Considering $j=-1/2$ and $i=\pm 1/2$, $\xi^2_0\ge 0$ implies that	$\xi_0^2=0$ and 
$\xi_0^3-\left[\xi_0^4+\xi_1^4(n)\right]/n=0$ for all $n\in\bbn$. Hence, $\xi_0^4=-\xi_1^4(n)$ for all $n\in\mathbb{N}$ and $\xi_0^3=0$. Consequently, we also have $\xi_2^3(n,m,i)=\widetilde{\xi}_2^3(n,m,i)=0$. But, then we have shown $w=0$, which is tantamount to $\mathcal{A}_0^2\cap - \widetilde{\mathcal{A}}_2^3=\{0\}\subseteq\widetilde{\mathcal{A}}_2^3$.

The same arguments apply for $t=3$ and thus  $\mathcal{A}_0^3\cap-\widetilde{\mathcal{A}}_3^3=\{0\}\subseteq \widetilde{\mathcal{A}}_3^3$.
\end{Example}

\begin{Remark}\label{26.11.2018.1}
Example~\ref{example:Strengthofthecondition} also allows us to discuss the following related question: Does the existence of a bid-ask process $(\widehat{\Pi}_t)_{t=0}^T$ with $\widehat{\Pi}_t\leq\Pi_t$ a.s. for all $t=0,\ldots,T$ such that $(\widehat{\Pi}_t)_{t=0}^T$ satisfies (NA) and 
	\begin{align}\label{eq:AdjustedImplicationNAr}
	\sum_{t=0}^{T}\xi_t=0\ \mathrm{a.s.\ with}\ \xi_t\in L^0(-K_t,\mathcal{F}_t)\Rightarrow \xi_t\in L^0(\wh{K}^0_t,\mathcal{F}_t)\ \mbox{for all\ } t=0,\ldots,T
	\end{align} 
	already imply the absence of an approximate arbitrage, i.e. $\overline{\mathcal{A}_0^T}\cap L^0(\mathbb{R}^d_+)=\{0\}$? Hereby \eqref{eq:AdjustedImplicationNAr} means that each transaction which is involved in a null-strategy in the original model is carried out at frictionless prices in the adjusted market. 
	
It turns out that the answer to the question is negative. Indeed, in the setting of 
Example~\ref{example:Strengthofthecondition}, we define the adjusted bid-ask process $(\widehat{\Pi}_t)_{t=0}^3$ by 
	\begin{align*}
	\widehat{\Pi}_0=
	\begin{pmatrix}
	1 & 1 & 1& 1\\
	a & 1 & \cdot & \cdot\\
	a & \cdot & 1& \cdot\\
	1 & \cdot& \cdot & 1 
	\end{pmatrix}, \quad
	\widehat{\Pi}_1=
	\begin{pmatrix}
	1 & a & a & 1\\
	a & 1 & \cdot & \cdot\\
	a & \cdot & 1& \cdot\\
	1 & \cdot& \cdot & 1 
	\end{pmatrix},
	\end{align*}
	$\widehat{\Pi}_2=\Pi_2$ and $\widehat{\Pi}_3=\Pi_3$. Then, we have $\widehat{\Pi}_t\leq\Pi_t$ a.s. for all $t=0,1,2,3$. In addition, the previously considered adjusted bid-ask process $(\widetilde{\Pi}_t)_{t=0}^3$ satisfies (NA) and yields better terms of trade than $(\widehat{\Pi}_t)_{t=0}^3$, i.e., $\widehat{\mathcal{A}}_0^T\subseteq\widetilde{\mathcal{A}}_0^T$. Consequently, $(\widehat{\Pi}_t)_{t=0}^3$ inherits the (NA) property. In addition, the only (up to multiplication with non-negative scalars) null strategy in the original market is $(\vt_t-\vt_{t-1})_{t=0}^3=(e^4-e^1,e^1-e^4,0,0)$. Thus, condition \eqref{eq:AdjustedImplicationNAr} is satisfied, but, as we have shown, $\overline{\mathcal{A}_0^T}\cap L^0(\mathbb{R}^d_+)\supsetneq\{0\}$.
	
The key observation is that the bid-ask process~$(\widehat{\Pi}_t)_{t=0}^3$ satisfies  
(\ref{eq:AdjustedImplicationNAr}) but not (\ref{eq:NArImplication}) since by the extension of the market from $(\Pi_t)_{t=0}^3$ to $(\widehat{\Pi}_t)_{t=0}^3$
new null-strategies occur. Indeed, it was a finding by 
Jacka, Berkaoui, and Warren~\cite{jacka2008no} (see the paragraph before Definition~3.2 therein) that it is not sufficient that 
the null-strategies in the original market are frictionless in the market extended by their ``adjusted trading prices'', i.e., it is not sufficient that (\ref{eq:AdjustedImplicationNAr}) is satisfied. However, in their Example~3.3 with $T=1$, a price adjustment~$\wh{\Pi}$
such that \eqref{eq:AdjustedImplicationNAr} holds true turns out to be sufficient 
to obtain a closed set of attainable portfolio values
(indeed, in Example~3.3 of \cite{jacka2008no}, 
the price adjustment essentially consists of adding $L^0(\mathrm{cone}(e^1-e^2),\mathcal{F}_1)$, which
is already necessary to make the null-strategy~$(e^1-e^2,e^2-e^1)=(e^2-e^1, e^4/\omega-e^1+e^3/\omega-e^4/\omega+e^2-e^3/\omega)$ of the original market frictionless in the extended market). 
On the other hand, Example~\ref{example:Strengthofthecondition} highlights 
why in the approach of \cite{jacka2008no} it would not be sufficient to 
ensure that  the null-strategies in the original market are frictionless in the extended market.
\end{Remark}

\bibliography{LiteraturNAp}

\begin{thebibliography}{DVKS07}

\bibitem[AB06]{aliprantis2006infinite}
C.~Aliprantis and K.~Border.
\newblock {\em Infinite Dimensional Analysis: A Hitchhiker's Guide}.
\newblock Springer-Verlag, 2006.

\bibitem[Bou06]{bouchard2006no}
B.~Bouchard.
\newblock No-arbitrage in discrete-time markets with proportional transaction
  costs and general information structure.
\newblock {\em Finance and Stochastics}, 10(2):276--297, 2006.

\bibitem[DMW90]{dalang1990equivalent}
R.~Dalang, A.~Morton, and W.~Willinger.
\newblock Equivalent martingale measures and no-arbitrage in stochastic
  securities market models.
\newblock {\em Stochastics: An International Journal of Probability and
  Stochastic Processes}, 29(2):185--201, 1990.

\bibitem[DS94]{delbaen1994general}
F.~Delbaen and W.~Schachermayer.
\newblock A general version of the fundamental theorem of asset pricing.
\newblock {\em Mathematische annalen}, 300(1):463--520, 1994.

\bibitem[DS06]{delbaen2006mathematics}
F.~Delbaen and W.~Schachermayer.
\newblock {\em The mathematics of arbitrage}.
\newblock Springer-Verlag, 2006.

\bibitem[DVKS07]{de2007no}
D.~De~Valli{\`e}re, Y.~Kabanov, and C.~Stricker.
\newblock No-arbitrage criteria for financial markets with transaction costs
  and incomplete information.
\newblock {\em Finance and Stochastics}, 11(2):237--251, 2007.

\bibitem[Gri05]{grigoriev2005low}
P.~Grigoriev.
\newblock On low dimensional case in the fundamental asset pricing theorem with
  transaction costs.
\newblock {\em Statistics \& Decisions}, 23(1):33--48, 2005.

\bibitem[HP81]{harrison1981martingales}
J.~Harrison and S.~Pliska.
\newblock Martingales and stochastic integrals in the theory of continuous
  trading.
\newblock {\em Stochastic processes and their applications}, 11(3):215--260,
  1981.

\bibitem[JBW08]{jacka2008no}
S.~Jacka, A.~Berkaoui, and J.~Warren.
\newblock No arbitrage and closure results for trading cones with transaction
  costs.
\newblock {\em Finance and Stochastics}, 12(4):583--600, 2008.

\bibitem[KRS02]{Kabanov2002}
Y.~Kabanov, M.~R{\'a}sonyi, and C.~Stricker.
\newblock No-arbitrage criteria for financial markets with efficient friction.
\newblock {\em Finance and Stochastics}, 6(3):371--382, 2002.

\bibitem[KRS03]{kabanov2003closedness}
Y.~Kabanov, M.~R{\'a}sonyi, and C.~Stricker.
\newblock On the closedness of sums of convex cones in {$L^0$} and the robust
  no-arbitrage property.
\newblock {\em Finance and Stochastics}, 7(3):403--411, 2003.

\bibitem[KS99]{kramkov1999}
D.~Kramkov and W.~Schachermayer.
\newblock The asymptotic elasticity of utility functions and optimal investment
  in incomplete markets.
\newblock {\em Ann. Appl. Probab.}, 9(3):904--950, 1999.

\bibitem[KS01a]{kabanov2001teacher}
Y.~Kabanov and C.~Sticker.
\newblock A teacher’s note on no-arbitrage criteria.
\newblock In {\em S{\'e}minaire de Probabilit{\'e}s XXXV}, pages 149--152.
  Springer, 2001.

\bibitem[KS01b]{kabanov2001harrison}
Y.~Kabanov and C.~Stricker.
\newblock The {H}arrison--{P}liska arbitrage pricing theorem under transaction
  costs.
\newblock {\em Journal of Mathematical Economics}, 35(2):185--196, 2001.

\bibitem[KS09]{kabanov2009markets}
Y.~Kabanov and M.~Safarian.
\newblock {\em Markets with transaction costs: Mathematical Theory}.
\newblock Springer-Verlag, 2009.

\bibitem[K{\"u}h18]{kuhn2018local}
C.~K{\"u}hn.
\newblock How local in time is the no-arbitrage property under capital gains
  taxes?
\newblock {\em To appear in Mathematics and Financial Economics}, 2018.

\bibitem[LZ]{jun.lepinette}
E.~L{\'{e}}pinette and J.~Zhao.
\newblock A complement to the {G}rigoriev theorem for the {K}abanov model.
\newblock {\em Preprint}.

\bibitem[Pen01]{Penner}
I.~Penner.
\newblock {Arbitragefreiheit in Finanzm\"arkten mit Transaktionskosten}.
\newblock Diplomarbeit, Humboldt-Universität zu Berlin, 2001.

\bibitem[Rok08]{rokhlin2008constructive}
D.~Rokhlin.
\newblock Constructive no-arbitrage criterion under transaction costs in the
  case of finite discrete time.
\newblock {\em Theory of Probability \& Its Applications}, 52(1):93--107, 2008.

\bibitem[RW09]{rockafellar2009variational}
R.~Rockafellar and R.~Wets.
\newblock {\em Variational analysis}, volume 317.
\newblock Springer-Verlag, 2009.

\bibitem[Sch92]{schachermayer1992hilbert}
W.~Schachermayer.
\newblock A {H}ilbert space proof of the fundamental theorem of asset pricing
  in finite discrete time.
\newblock {\em Insurance: Mathematics and Economics}, 11(4):249--257, 1992.

\bibitem[Sch04]{schachermayer2004fundamental}
W.~Schachermayer.
\newblock The fundamental theorem of asset pricing under proportional
  transaction costs in finite discrete time.
\newblock {\em Mathematical Finance}, 14(1):19--48, 2004.

\end{thebibliography}
\end{document}